%% file: quickselect-paper-arXiv.tex
\documentclass[a4paper,final,11pt,american,DIV=11,includehead]{scrartcl}
\usepackage{lmodern} % for teletype font
\usepackage[utf8]{inputenc}
\usepackage[T1]{fontenc}

\input{quickselect-paper-preamble-arXiv}

\def\email#1{\texttt{#1}}

\title{Analysis of Quickselect under\\ Yaroslavskiy's Dual-Pivoting Algorithm}
\author{Sebastian Wild${}^*$
\and Markus E. Nebel\thanks{%
	Computer Science Department,
	University of Kaiserslautern,
	Germany
	\email{\{wild,nebel\}\,@cs.uni-kl.de}
}
\and Hosam Mahmoud\thanks{%
	Hosam Mahmoud,
	Department of Statistics,
	The George Washington University,
	Washington, D.C. 20052, U.S.A.
}}

\begin{document}  
\maketitle

\begin{abstract}
\setlength\parindent{1.5em}
There is excitement within the algorithms community about a new
partitioning method introduced by Yaroslavskiy. 
This algorithm renders Quicksort slightly faster than the case when it runs under 
classic partitioning methods. %with respect the number of comparisons. 
We show that this improved performance in Quicksort is \emph{not} sustained in
Quickselect; a variant of Quicksort for finding order statistics.

We investigate the number of comparisons
made by Quickselect to find a key with a randomly selected
rank under Yaroslavskiy's algorithm. 
This \weakemph{grand averaging} is a smoothing operator
over all individual distributions for
specific fixed order statistics. 
We give the exact grand average. 
The grand distribution of the number of comparison (when suitably scaled) 
is given as the fixed-point  
solution of a distributional equation of a
contraction in the Zolotarev metric space.   
Our investigation shows that Quickselect under 
older partitioning methods
slightly outperforms Quickselect under Yaroslavskiy's algorithm, 
for an order statistic of a random rank. 
Similar results are obtained for extremal order statistics, where again we find
the exact average, and the distribution for the number of comparisons
(when suitably scaled). 
Both limiting distributions are of perpetuities
(a sum of products of independent mixed continuous 
random variables). 
\bigskip

\noindent\textbf{AMS subject classifications:} Primary:
%\subclass{Primary:
   60C05;    % Combinatorial probability
 secondary:
   68P10,    % Searching and sorting
   68P20.    % Information storage and retrieval
\smallskip{}

\end{abstract}

\section{Quicksort, What Is New?}

Quicksort is a classic fast sorting algorithm.
It was originally published by Hoare~\cite{Hoare1962}.
Quicksort is the method of choice to implement a sorting function in many
widely used program libraries, 
e.\,g.\ in the C/C++ standard library and in the Java runtime environment. 
The algorithm is recursive, and at each level of recursion
it uses a partitioning algorithm.
Classic implementations of Quicksort use a variety of partitioning methods
derived from fundamental versions invented by Hoare~\cite{Hoare1961,Hoare1962}
and refined and popularized by 
Sedgewick; see for example~\cite{Sedgewick1977}.

Very recently, a partitioning algorithm proposed by Yaroslavskiy created
some sensation in the algorithms community. The excitement arises
from various indications, theoretical and experimental, that
Quicksort on average runs faster under Yaroslavskiy's dual pivoting 
(see~\cite{Wild2012}).
Indeed, after extensive experimentation Oracle adopted Yaroslavskiy's
dual-pivot Quicksort as the default sorting method for their Java~7 
runtime environment, a software platform used on many computers worldwide.

\section{Quicksort and Quickselect}
Quicksort is a two-sided algorithm for sorting data (also called \emph{keys}). 
In a classic implementation, it puts a pivot key in its correct
position, and arranges the data in two groups relative to that pivot.
Keys smaller than the pivot are put in one group, the rest are placed in the other group.
The two groups are then sorted recursively.

The one-sided version (Quickselect) of the algorithm can be used to find
order statistics. 
It is also known as Hoare's ``Find'' algorithm, which was first given
in~\cite{Hoare1961}.
To find a certain order statistic, such as the first quartile,
Quickselect goes through the partitioning stage, just as in Quicksort,  
then the algorithm decides whether the pivot is the sought
order statistic or not. If it is, the algorithm terminates 
(announcing the pivot to be the sought element); 
if not, it recursively pursues only the group on the side (left or right of
the pivot) where the order statistic resides.
We know which side to choose, as the rank of the pivot becomes known 
after partitioning.

There are algorithms for specific order statistics
like smallest, second smallest, largest, median, etc. 
However, adapting one of them to work for a different order 
statistic is not an easy task. 
Take for example the algorithm for the second largest, 
and suppose we want to tinker with it to find the median. 
It cannot be done without entirely rewriting the algorithm. 
On the other hand, Quickselect is one algorithm that is versatile 
enough to deal with \emph{any} rank without changing any statements in it,
a feature that may be appealing in practice, particularly
for a library function that cannot predict which 
order statistic will be sought. 

A standard measure for the analysis of a comparison-based sorting
algorithm is the number
of \emph{data} comparisons it makes while sorting; 
see for example~\cite{Kirschenhofer1998,Knuth1998}. 
Other types of comparison take place while sorting, such as index
or pointer comparisons. However, they are negligible in view
of the fact that they mostly occur at lower asymptotic orders,
and any individual one of them typically
costs considerably less than
an individual data comparison. For instance,
comparing two indices is a comparison of two short integers,
while two keys can be rather long such as business records, 
polynomials,  or DNA strands, typically each
comprising thousands of nucleotides. 
Hence, these additional index and pointer comparisons are often
ignored in the analysis. We shall follow this tradition. 

Unless further information on the input data is available (e.\,g., in some
specialized application), one strives for an analysis that remains as generic as
possible while still providing sensible predictive quality for use
cases in practice.
Therefore, we consider the \emph{expected} costs when operating on an input
chosen uniformly at random, which for selection means that all $n!$ orderings of $n$
keys are equally likely.
(This assumption can be enforced for any input if pivots are
chosen randomly.)

We furthermore assume that the sought order statistic is also chosen uniformly
at random among all $n$ possible ranks.
This models the scenario of a generic library function
intended to cope with any possible selection task.
We note that parts of the analysis remain feasible when the sought rank $r$ is 
%at a fixed proportion $\rho=r/n$ from $n$, where $\rho$ is then 
kept fixed (becoming a second parameter of the
analysis)~\cite{Kirschenhofer1998,Knuth1998}.
However, for \emph{comparing} different versions of Quickselect, a
one-dimensional measure is much easier to interpret and in light of our clear
(and negative) results, the parametric analysis is
unlikely to provide additional algorithmic insights.%
\footnote{
	To shed some light on that, we computed the expected comparison counts for
	$n=200$ and $n=300$ with fixed $r$ for all ranks $1\le r\le n$.	
	We find that Yaroslavskiy's algorithm needs more comparisons in \emph{all}
	these cases. We see no reason to believe that this will change for
	larger~$n$.%
}

We do, however, think that a \emph{distributional} analysis is important.
Unlike for Quicksort, whose costs become more and more concentrated
around their expectation as $n$ increases, the standard deviation of
(classic) Quickselect is of the same order as the expectation, (both are
linear)~\cite{Kirschenhofer1998}.
This means that even for $n\to\infty$, substantial deviations from the
mean are to be expected.
As the use of more pivots tends to give more balanced subproblem sizes, 
it seems plausible that these deviations can be reduced by switching to the
dual-pivot scheme described below. 
Therefore, we also compute the variance and a limit distribution for the
number of comparisons.

It is worth
mentioning that a fair contrast between 
comparison-based sorting algorithms and sorting algorithms based on
other techniques (such as radix selection, which uses comparisons of bits)
should resort to the use of one basis, such as how many bits are compared
in both. Indeed, in comparing two very long bit strings, we can decide almost
immediately that the two strings are different, if they differ in the 
first bit. In other instances, where the two strings are ``similar,'' 
we may run a very long sequence of bit comparisons till we discover
the difference. Attention  to this type of contrast is taken up in 
\cite{Fill2013} and other sources.

Other associated cost measures include
the number of swaps or data moves \cite{Mahmoud2010,Martinez2009}. 
We do not discuss those in detail in this paper, but we note that
swaps can be analyzed in a very similar way (reusing the
number of swaps in one partitioning step from our previous work on Quicksort
\cite{Wild2013Quicksort}).
The resulting expected values are reported in Table~\ref{tab:results}.

\section{Dual Pivoting}

The idea of using two pivots (\weakemph{dual-pivoting}) had been suggested
before, see Sedgewick's and Hennequin's Ph.\,D.\
dissertations~\cite{hennequin1991analyse,Sedgewick1975}. 
Nonetheless, the implementations
considered at the time did not show any promise.
Analysis reveals that Sedgewick's dual-pivot Quicksort variant
performs an asymptotic
average of $\frac {32} {15} n\ln n + \Oh(n)$
data comparisons, while the classic (single-pivot) version uses only
an asymptotic average of $2 n \ln n + \Oh(n)$ comparisons
\cite{Sedgewick1975,Wild2012}. 
Hennequin's variant performs $2 n \ln n + \Oh(n)$ comparisons
\cite{hennequin1991analyse}\,---\,asymptotically the same as classic
Quicksort. However, the inherently more complicated dual-pivot partitioning
process is presumed to render it less efficient in practice.

These discoveries were perhaps a reason to discourage
research on sorting with multiple-pivot partitioning, till Yaroslavskiy
carefully heeded implementation details. 
His dual-partitioning algorithm improvement broke a psychological barrier.
Would such improvements be sustained in Quickselect? 
It is our aim in this paper to answer this question. 
We find out that there is no improvement in the number of
comparisons: 
Quickselect under Yaroslavskiy's dual-pivot partitioning
algorithm (simply Yaroslavskiy's algorithm, henceforth)
is slightly \emph{worse} than classic single-pivot Quickselect.

\smallskip
Suppose we intend to sort $n$ distinct keys stored in the array $A[1\range n]$.  
Dual partitioning uses \emph{two} pivots, as opposed to the single pivot used in
classic Quicksort. Let us assume the two pivots are initially
$A[1]$ and $A[n]$, and suppose their ranks are $p$ and $q$. If $p > q$,
we swap the pivots. 
While seeking two positions for the two pivots, 
the rest of the data is categorized in three groups: small, medium and large. 
Small keys are those with ranks less than~$p$, medium keys have ranks
at least $p$ and less than $q$, and large keys are those with ranks at least $q$.
Small keys
are moved to positions lower than $p$, large keys are
moved to positions higher than $q$, medium keys are
kept in positions between $p+1$ and $q-1$.
So, the two keys with ranks $p$ and $q$ can be moved to their correct
and final positions.

After this partitioning stage, dual-pivot Quicksort 
then invokes itself recursively (thrice) on $A[1\range p-1]$,  $A[p+1\range q-1]$ 
and $A[q+1\range n]$.
The boundary conditions are the very small arrays of size 0 (no keys
to sort), arrays of size 1 (such an array is already sorted), and
arrays of size 2 (these arrays need only one comparison between the two keys
in them); in these cases no further recursion is invoked.

\begin{algorithm}
	\vspace*{-1ex}
	\def\pind{\id{i_p}}
	\def\qind{\id{i_q}}
	\def\signum{\mathrm{sgn}}
	\def\case#1{$\kw{in case}\like[r]{MM}{#1}\;\kw{do}\;$}
	\begin{codebox}
		\Procname{$\proc{Quickselect}\,(\arrayA,\id{left},\id{right},r)$}
		\zi \Comment Assumes $\id{left} \le r \le \id{right}$. 
		\zi \Comment Returns the element that would reside in $\arrayA[r]$ after
		sorting $\arrayA[\id{left} \range \id{right}]$.
		\li \If $\id{right} \le \id{left}$
		\li \Then 
				\Return $\arrayA[\id{left}]$
		\li	\Else
		\li		$(\pind, \qind) \gets 
					\proc{PartitionYaroslavskiy}(\arrayA,\id{left}, \id{right})$
		\li		$c \gets \signum(r - \pind) \bin+ \signum(r - \qind)$
				\qquad\Comment Here $\signum$ denotes the signum function.
		\li		\kw{case distinction} on the value of $c$
		\li		\Do
		 			\case{-2}	\Return $\proc{Quickselect}\,(\arrayA,
									\makeboxlike[c]{$\pind+1$}{$\id{left}$},
									\makeboxlike[c]{$\pind+1$}{$\pind-1$},
									r
								)$
								\label{lin:yaroslavskiy-call-1}
		\li			\case{-1} 	\Return $\arrayA[\pind]$
		\li			\case{0}	\Return $\proc{Quickselect}\,(\arrayA,
									\makeboxlike[c]{$\pind+1$}{$\pind+1$},
									\makeboxlike[c]{$\pind+1$}{$\qind-1$},
									r
								)$
								\label{lin:yaroslavskiy-call-2}
		\li			\case{+1}	\Return $\arrayA[\qind]$
		\li			\case{+2} 	\Return $\proc{Quickselect}\,(\arrayA,
									\makeboxlike[c]{$\pind+1$}{$\qind+1$},
									\makeboxlike[c]{$\pind+1$}{$\id{right}$},
									r
								)$
								\label{lin:yaroslavskiy-call-3}
				\End
		\li		\kw{end cases}
			\EndIf
		\zi
	\end{codebox}
	\vspace{-4ex}	
	\caption{\strut%
		Dual-pivot Quickselect algorithm for finding the $r$th order statistic.
	}
	\label{Alg:osqs}
\end{algorithm}

\begin{algorithm}
\vspace{-1ex}
	\def\pind{\id{i_p}}
	\def\qind{\id{i_q}}
	\begin{codebox}
		\Procname{$\proc{PartitionYaroslavskiy}\,(\arrayA,\id{left},\id{right})$}
		\zi \Comment Assumes $\id{left} \le \id{right}$.
		\zi \Comment Rearranges \arrayA such that with $(\pind,\qind)$ the return
				value holds	\smash{ $\begin{cases}
	 				\forall \: \id{left}\le j \le \pind	& \arrayA[j] \le p \\
	 				\forall \:\id{\pind}\le j \le \qind	& p \le \arrayA[j] \le q \\
	 				\forall \: \qind\le j \le \id{right}	& \arrayA[j] \ge q
				\end{cases}$}\;.
				\rule[-2ex]{0pt}{1ex}
		\li	\If $\arrayA[\id{left}] > \arrayA[\id{right}]$
					\label{lin:yaroslavskiy-comp-0} 
		\li	\Then
				$p\gets \arrayA[\id{right}]$; \>\>\>\> $q\gets \arrayA[\id{left}]$
		\li	\Else
		\li		$p\gets \arrayA[\id{left}]$;  \>\>\>\> $q\gets \arrayA[\id{right}]$
			\EndIf
		\li $\ell\gets \id{left} + 1$; 
		 	\quad $g\gets \id{right} - 1$; 
		 	\quad $k\gets \ell$ \label{lin:yaroslavskiy-init-l-g-k} 
		\li	\While $k\le g$ 
		\li	\Do
				\If $\arrayA[k] < p$ \label{lin:yaroslavskiy-comp-1}
		\li		\Then
					Swap $\arrayA[k]$ and $\arrayA[\ell]$ \label{lin:yaroslavskiy-swap-1}
		\li			$\ell\gets \ell+1$ \label{lin:yaroslavskiy-l++-1}
		\li		\Else 
		\li			\If $\arrayA[k] \ge q$ \label{lin:yaroslavskiy-comp-2}
		\li			\Then
						\While $\arrayA[g] > q$ and $k<g$ \kw{do} $g\gets g-1$ \kw{end while} \label{lin:yaroslavskiy-comp-3}
		\li				\If $\arrayA[g] \ge p$ \label{lin:yaroslavskiy-comp-4}
		\li				\Then
							Swap $\arrayA[k]$ and $\arrayA[g]$ \label{lin:yaroslavskiy-swap-2}
		\li				\Else
		\li					Swap $\arrayA[k]$ and $\arrayA[g]$; \;
							Swap $\arrayA[k]$ and $\arrayA[\ell]$ \label{lin:yaroslavskiy-swap-3}
		\li					$\ell\gets \ell+1$ \label{lin:yaroslavskiy-l++-2}
						\EndIf
		\li				$g\gets g-1$ \label{lin:yaroslavskiy-g--}
					\EndIf
				\EndIf
		\li		$k\gets k+1$ \label{lin:yaroslavskiy-k++}
			\EndWhile \label{lin:yaroslavskiy-end-while}
		\li	$\ell\gets \ell-1$; \>\>\>$g\gets g+1$
		\li	$\arrayA[\id{left}] \gets \arrayA[\ell]$; \>\>\>\>\> $\arrayA[\ell] \gets p$
				\label{lin:yaroslavskiy-swap-4} 
				\qquad\Comment Swap pivots to final positions 
		\li	$\arrayA[\id{right}] \gets \arrayA[g]$; \>\>\>\>\> $\arrayA[g] \gets q$
				\label{lin:yaroslavskiy-swap-5} 
		\li	\Return $(\ell, g)$
		\zi
	\end{codebox}
	\vspace{-4ex}
	\caption{\strut%
		Yaroslavskiy's dual-pivot partitioning algorithm.
	}
\label{Alg:Yaroslavskiy}
\end{algorithm}

This is the general paradigm for dual pivoting. 
However, it can be implemented in many different ways. 
Yaroslavskiy's algorithm keeps track of three pointers: 
\begin{itemize}
\item $\ell$,  moving up from lower to higher indices, and below which
all the keys have ranks less than $p$.
\item $g$, moving down from higher to lower indices, and above which
all the keys have ranks at least $q$.
\item $k$, moving up beyond $\ell$ and not past $g$. During the execution,
all the keys at or  below position $k$ and above $\ell$ are medium,
with ranks lying between $p$ and $q$ (both inclusively).
\end{itemize}  

Hence, the three pointers $\ell$, $k$ and $g$ divide the array into four
ranges, where we keep the relation of elements invariantly as given above.
Graphically, this reads as follows:
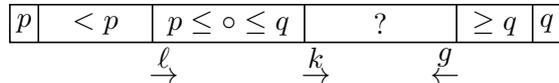
\begin{quote}
	\mbox{}\hfill%
	\begin{tikzpicture}[
		scale=0.5,
		baseline=(ref.south),
		every node/.style={font={}},
		semithick,
	]	
	
	\draw (-.75,0) -- ++(14.5,0) -- ++(0,1) -- ++(-14.5,0) -- cycle;
	\node at (-.375, .5) {$p$} ;
	\node at (13.375, .5) {$q$} ;
	\draw (0,0) -- ++(0,1);
	\draw (13,0) -- ++(0,1);
	
	\node at (1.5,0.5) {$< p$};
	%\draw (2,1) -- ++(0,-1);
	\draw (3,1) -- ++ (0,-1);
	\node at (3.3,-0.4) {$\ell$};
	
	\node at (12,0.5) {$\ge q$};
	%\draw (10.25,1) -- ++(0,-1);
	\draw (11,1) -- ++ (0,-1);
	\node at (10.7,-0.4) {$g$};
	
	\node at (5,0.5) {$p\le \circ\le q$};
	\draw (7,1) -- ++(0,-1);
	%\draw (8,1) -- ++ (0,-1);
	\node at (7.3,-0.4) {$k$};

 	\begin{pgfinterruptboundingbox}
	\node[below] at (10.7,-0.5) {$\leftarrow$};
	\node[below] at (3.3,-0.5) {$\rightarrow$};
	\node[below] at (7.3,-0.5) {$\rightarrow$};
 	\end{pgfinterruptboundingbox}
	
	\node[inner sep=0pt] (ref) at (9,0.5) {?};
	\end{tikzpicture}%
	\hfill\mbox{}
\end{quote}
\medskip

The adaptation of Quicksort to deliver a certain order statistic $r$ 
(the $r$th smallest key) is straightforward. Once the positions
for the two pivots are determined,
we know whether $r < p$, $r = p$, $p < r < q$, $r=q$,
or $r > q$. 
If $r = p$ or $r =q$, the algorithm declares
one of the two pivots (now residing at position $p$ or $q$) 
as the required $r$th order statistic, and terminates. 
If $r \ne p$ and $r \ne q$, 
the algorithm chooses one of three subarrays: 
If $r < p$ the algorithm recursively seeks the $r$th order statistic in
\mbox{$A[1\range p-1]$}, if $p < r < q$, the algorithm seeks the
$(r-p)$th order statistic among the keys of $A[p+1\range q-1]$;
this $(r-p)$th key is, of course, ranked $r$th in the entire data set.
If $r > q$, the algorithm seeks the $(r-q)$th order statistics in $A[q+1 \range n]$;
this $(r-q)$th element
is then ranked $r$th in the entire data set.
Thus, only the subarray containing the desired order statistic is searched and the 
others are ignored. 

Algorithm~\ref{Alg:osqs} is the formal algorithm in pseudo code. The code calls
Yaroslavskiy's dual-pivot partitioning procedure 
(given as Algorithm~\ref{Alg:Yaroslavskiy}). 
The code is written to work on the general subarray $A[\id{left} \range
\id{right}]$ in later stages of the recursion, and the initial call is 
$\proc{Quickselect}(A,1,n,r)$.

Note that in Algorithm~\ref{Alg:Yaroslavskiy}, variables $p$ and~$q$ are used
to denote the data elements used as pivots, whereas in the main text, $p$
and~$q$ always refer to the \emph{ranks} of these pivot elements relative to the
current subarray.
We kept the variables names in the algorithm to stay consistent with the
literature.

A few words are in order to address the case of \emph{equal elements}.
If the input array contains equal keys, several competing notions of
ranks exist.
We choose an \weakemph{ordinal ranking} that fits our situation best: 
The rank of an element is defined as its \emph{index} in the array after
sorting it with (a corresponding variant of) Quicksort.
With this definition of ranks,  
Algorithm~\ref{Alg:osqs} correctly handles arrays with equal
elements.

\section{Randomness Preservation}
\label{pro:randomness-preservation}

We recall from above that our probability model on data assumes the keys 
to be in random order (\textsl{random permutation model}). 
Since only the relative ranking is important, 
we assume w.\,l.\,o.\,g.\ (see \cite{Knuth1998}) that the $n$ keys are  
real numbers independently sampled from a common \emph{continuous}
probability distribution.

Several partitioning algorithms can be employed;
a good one produces subarrays again following the random
permutation model in subsequent recursive steps:
%
%\begin{quote}
	{\slshape
	If the whole array is a (uniformly chosen) random permutation of its
	elements, so are the subarrays produced by partitioning.}
%\end{quote}

For instance, in classic single-pivot Quicksort,
if $p$ is the final position of the pivot, 
then right after the first partitioning stage the relative ranks of 
$A[1], \ldots, A[p-1]$ are a random permutation of $\{1,\ldots, p-1\}$ 
and the relative ranks of $A[p+1], \ldots, A[n]$ are a random permutation of 
$\{1,\ldots, n-p\}$, see \cite{hennequin1989combinatorial} or~\cite{Knuth1998}.

Randomness preservation enhances performance on random data,
and is instrumental in formulating recurrence equations for the analysis.
Hoare's \cite{Hoare1962} and Lomuto's \cite{Bentley1984}
single-pivot partitioning algorithms are 
known to enjoy this important and desirable property.

\begin{lemma}
Yaroslavskiy's algorithm (Algorithm~\ref{Alg:Yaroslavskiy}) 
is randomness preserving.
\end{lemma}
\begin{proof}
Obviously, every key comparison in Yaroslavskiy's algorithm involves (at least)
one pivot element; see lines~\ref*{lin:yaroslavskiy-comp-0},
\ref*{lin:yaroslavskiy-comp-1}, \ref*{lin:yaroslavskiy-comp-2}, \ref*{lin:yaroslavskiy-comp-3} and
\ref*{lin:yaroslavskiy-comp-4} of Algorithm~\ref{Alg:Yaroslavskiy}. 
Hennequin shows that this is a sufficient
criterion for randomness preservation \cite{hennequin1989combinatorial}, so
Yaroslavskiy's algorithm indeed creates random subarrays.
\end{proof}

\section{Main Results}

We investigate the performance of Quickselect's number of data comparisons, when
it seeks a key of a randomly selected rank, 
while employing Yaroslavskiy's algorithm.
The exact grand average number of data comparisons 
is given in the following statement,
in which
$H_n$ is the $n$th harmonic number $\sum_{k=1}^n 1/k$.

\begin{proposition}
\label{Prop:grandave}
	Let $C_n$ be the number of data comparisons exercised while Quickselect is searching
	under Yaroslavskiy's algorithm
	for an order statistic chosen uniformly at random from all possible ranks. 
	For $n \ge 4$,
	\begin{align*}
			\E[C_n] 
		&\wwrel=   
			\tfrac {19} 6 n  - \tfrac {37} 5  H_n  +\tfrac{1183}{100} 
			- \tfrac {37}{5} H_n n^{-1} - \tfrac {71} {300} n^{-1}  
	     \wwrel\sim \tfrac {19} 6 n.
	\end{align*}
\end{proposition}

We use the notation $\, \eqlaw\, $ to mean (exact) equality in distribution,
and $\convD$ to mean weak convergence in distribution.

\begin{theorem}
\label{Theo:randomrank}
	Let $C_n$ be the number of comparisons made by Quickselect with 
	Yaroslavskiy's algorithm while searching
	for an order statistic chosen uniformly at random from all possible ranks.
	The random variables  $C^*_n \ce C_n / n$ converge in distribution
	and in second moments to a limiting random variable~$C^*$ that satisfies the
	distributional equations
	\begin{align}
			C^*  
		&\wwrel\eqlaw 
		  		U_{(1)} \indicator_{\{V < {U_{(1)}}\}} \, C^*
				\wbin+ 	\bigl(U_{(2)} - U_{(1)}\bigr) 
					\indicator_ {\{U_{(1)} < V < U_{(2)}\}} \,
						\Csp
			\label{eq:Cstar-direct-limit-equation}
		\\*&\wwrel\ppe {} 
				\wbin+	\bigl(1-U_{(2)}\bigr) \indicator_ {\{V > U_{(2)}\}} 
				\, \Cspp \wwbin+ 1 + U_{(2)} \bigl( 2 - U_{(1)} - U_{(2)} \bigr) \,,
			\nonumber
		\\		
				C^*
			&\wwrel\eqlaw 
			  X^* C^* + g(X^*, W^*) \;.
		\label{Eq:Cstar}
	\end{align}         
	where $\Csp$ and $\Cspp$ are independent copies of $C^*$, which are also
	independent of $(U_{(1)},U_{(2)},V)$ and $(X^*, W^*)$; 
	$(U_{(1)},U_{(2)})$ are the order statistics of two independent $\uniform(0,1)$
	random variables, $V$ is a $\uniform(0,1)$ random variable 
	independent of all else, and  $(X^*, W^*)$ have a bivariate density
	$$
		f(x,w) \wwrel= \begin{cases} 
						6x, &\text{for } 0 < x <  w < 1;\cr
	                     0, &\text{elsewhere},
	           \end{cases} 
	$$
	and $g(X^*, W^*)$ is a fair mixture%
	\footnote{% 
		A fair mixture of three random variables is obtained by first choosing one of
		the three distributions at random, all three being equally likely, then
		generating a random variable from that distribution.
	} 
	of the three random variables
    $$
    	1 + W^*(2-X^*-W^*), 
    	\quad 1+(1+X^*-W^*)(2W^*-X^*), 
    	\quad 1+(1-X^*)(X^*+W^*) \;.
    $$
\end{theorem}

As a corollary of Theorem~\ref{Theo:randomrank}, 
we find that 
$\V[C_n^*] \sim  \frac {25} {36} n^2 = 0.69\overline4 n^2$, 
as \mbox{$n \to \infty$}.
Another corollary is to write $C^*$ explicitly as a sum of
products of independent random variables: 
\begin{align*}
		 C^*
	&\wwrel\eqlaw 
		\sum_{j=1}^\infty g(X_j, W_j)\Bigl(\: \prod_{k=1}^{j-1} X_j \Bigr),
	\qquad\text{where}\qquad X_j \wwrel\eqlaw X^*,
	\qquad W_j \wwrel\eqlaw W^*,
\end{align*} 
and $\{X_j\}_{j=1}^\infty$ is a family of totally independent random
variables,%
\footnote{%
	For the usual definition of total independence see any
	classic book on probability, such as \cite[p.\,53]{Chung2001}, for example.
} 
and so is 
$\{W_j\}_{j=1}^\infty$.

\bigskip\noindent
\textbf{Remark} {\sl
	Let $\{V_j\}_{j=1}^\infty$ and $\{Z_k\}_{k=1}^\infty$ be two families of
	totally independent random variables, and assume~$V_j$ is independent of~$Z_k$
	for each $j, k \ge 1$. Sums of products of independent random variables of the form
	$$ V_1 + V_2 Z_1 + V_3 Z_1 Z_2 + V_4 Z_1 Z_2 Z_3 + \cdots$$
	are called perpetuities. They appear in  
	financial mathematics~\cite{Embrechts1997},  
	in stochastic recursive	algorithms~\cite{Iksanov2009}, 
	and in many other areas.
\/}

\smallskip
\begin{proposition}
\label{Prop:extremalave}
	Let $\hat C_n$ 
	be the number of comparisons made by
	Quickselect with Yaroslavskiy's algorithm to find the smallest key of a random
	input of size~$n$. We then have
	\begin{align*} 
			\E\bigl[\hat C_n\bigr]
		&\wwrel=  
			\frac 1 {24n(n-1)(n-2)} \Bigl(
					57n^4 - 48n^3 H_n - 178 n^3 + 144 n^2 H_n 
		\\*&\qquad\qquad {}  
					+ 135 n^2 -96nH_n - 14 n + 24\Bigr) \,,
		\qquad \text{for } n \ge 4 \,,
		\\[-.5ex]
	         &\wwrel\sim \tfrac {19} 8 n \;.
	\end{align*} 
\end{proposition}

\begin{theorem}
\label{Theo:extremal}
	Let $\hat C_n$
	be the number of comparisons made by Quickselect under Yaroslavskiy's
	algorithm on a random input of size~$n$ to find the smallest order 
	statistic. 
	The random variables $\hat C^* = \hat C_n  / n$ converge
	in distribution and in second moments 
	to a limiting random variable $\hat C^*$
	satisfying the distributional equation 
	\begin{align}
		\hat  C^* 
		&\wwrel\eqlaw 
			U_{(1)} \hat C^* 
			\wbin+ 1 + U_{(2)} \bigl( 2 - U_{(1)} - U_{(2)} \bigr),
	\label{Eq:Chat}
	\end{align} 
	where $U_{(1)}$ and $U_{(2)}$ are respectively the minimum and
	maximum of two independent random variables, both distributed uniformly on $(0,1)$.
\end{theorem}

As a corollary, we find for $\hat C_n$, as $n\to\infty$,
\begin{align*}
		\E[\hat C_n]
	&\wwrel\sim
		\tfrac{19}6 n, 
	\qquad\text{and}\qquad
		\V[\hat C_n] 
	\wwrel\sim 
		\tfrac {1261} {4800}\, n^2 
	\wwrel= 
		0.262708\overline3 n^2 \;.
\end{align*} 
Another corollary is that $\hat C^*$ can be written explicitly
as a sum of products of independent random variables: 
\begin{align*}
		\hat C^*
	&\wwrel\eqlaw 
		\sum_{j=1}^\infty Y_j\Bigl( \prod_{k=1}^{j-1} \hat X_j \Bigr),
       \quad\text{where}\quad 
       		\hat X_j \wrel\eqlaw U_{(1)}, \quad
			Y_j \wrel\eqlaw 1 + U_{(2)} \bigl(2-U_{(1)} - U_{(2)} \bigr),
\end{align*} 
here $\{\hat X_j\}_{j=1}^\infty$ is a family of totally independent random
variables, and so is  $\{Y_j\}_{j=1}^\infty$.

Similar results can be developed for $\check C_n$, the cost for 
finding the maximal order statistics.
We note that it is not exactly symmetrical with $\hat C_n$. 
For instance, while
$\check C_n$ has the same asymptotic mean as  $\hat C_n$, it has 
a different asymptotic variance, which is namely
$$\V[\check C_n] 
	\wwrel\sim 
		\tfrac {1717}{4800}\, n^2 
	\wwrel= 
		0.357708\overline3  n^2 \;.$$
%\begin{align*}
%		\bar C^*
%	&\wwrel\eqlaw 
%		\sum_{j=1}^\infty Y_j\Bigl( \prod_{k=1}^{j-1} \bar X_j \Bigr),
   %    \quad\text{where}\quad
      % 		\bar X_j \wrel\eqlaw 1 - U_{(2)}, \quad
	%		Y_j \wrel\eqlaw 1 + U_{(2)} \bigl(2-U_{(1)} - U_{(2)} \bigr),
%\end{align*} 
%where $\{\bar X_j\}_{j=1}^\infty$ is a another family of totally independent
%random variables.
In fact, similar results can be developed for the number of comparisons
needed for any extremal order statistic (very small or very large), that is 
when the sought rank $r$ is $o(n)$ or $n-o(n)$.

\section{Organization}

The rest of the paper is devoted to the proof and is organized as follows.
Section~\ref{Sec:notation} sets up fundamental components
of the analysis, and working notation that will be used throughout. 
In Section~\ref{Sec:Yaroslavskiy}, we present a probabilistic analysis
of Yaroslavskiy's algorithm. The analysis under rank smoothing
is carried out in Section~\ref{Sec:grand}, which has two subsections:
Subsection~\ref{Subsec:grandave} is for the exact grand average,
and Subsection~\ref{Subsec:granddistribution} is for the limiting
grand distribution via the contraction method.
We say a few words in that subsection on the
origin and recent developments of the method and its success 
in analyzing divide-and-conquer algorithms. 

In Section~\ref{Sec:extremal}, we import the methodology to obtain results for
extremal order statistics. 
Again, Section~\ref{Sec:extremal} has two subsections:
Subsection~\ref{Subsec:extremalave} is for the exact average,
and Subsection~\ref{Subsec:extremaldistribution} is for the limiting
distribution.
We conclude the paper in Section~\ref{Sec:remarks}
with remarks on the overall perspective of the use of Yaroslavskiy's
algorithm in Quicksort and Quickselect. 

The appendices are devoted to proving some technical points;
Appendix~\ref{app:definition-spacings} lays common foundations and
Appendices~\ref{app:proof-convergence-grand-avg}
and~\ref{app:proof-convergence-extremal} formally show convergence
to limit law for random ranks, respectively extremal ranks.
Finally, Appendix~\ref{app:proof-lemma-equivalantdistribution} proof a 
technical lemma showing that $C^*$ is distributed like a perpetuity.

\section{Preliminaries and Notation}
\label{Sec:notation}

We shall use the following standard notation: 
$\one {{\cal E}}$ is the \emph{indicator random variable} of the event $\cal E$
that assumes the value 1, when $\cal E$ occurs, and assumes the
value~0, otherwise. 
$\Prob(\mathcal E)$ is the probability that $\mathcal E$ occurs and
$\E[X]$ denotes the expected value of random variable $X$.
%We use the notation $\smash\eqlaw$ to mean (exact) \emph{equality in 
%distribution}, and $\smash\convD$ to
%mean \emph{weak convergence} (convergence in distribution). 

The notation
$\smash\almostsure$ stands for \emph{convergence almost surely}.
By $\|X\|_p \ce \E[|X|^p]^{1/p}$ with $1\le p<\infty$,
we denote the \emph{$L_p$-norm} of random variable
$X$, and we say random variables $X_1,X_2,\ldots$ \emph{converge in $L_p$ to
$X$}, shortly written as \smash{$X_n \convL{\smash{p}} X$}, when 
$\lim_{n\to\infty} \| X_n - X \|_p = 0$.
Unless otherwise stated, 
all asymptotic equivalents and bounds concern the limit as $n\to\infty$.

Let $\Hyper(n, s, w)$ be a \emph{hypergeometric} random variable; 
that is
the number of white balls in a size $s$ sample of balls taken at random without
replacement (all subsets of size $s$ being equally likely) from an urn
containing a total of $n$ white and black balls, of which $w$ are white. 
% The mean and variance for this standard distribution are given by the
% formulas
% \cite{Stuart1994}.
% %
% \begin{align}
% 		\E[\Hyper(n ,s, w)]
% 	&\wwrel= \frac {w s} n,
% 	\label{Eq:meanhyper} \\
% 		\V[\Hyper(n ,s, w)]
% 	&\wwrel= \frac {w s(n-w)(n-s)} {n^2 (n-1)} \;.
% 	\label{Eq:varHyper}
% \end{align}
% %

Let $P_n$ be the (random) \emph{rank} of the smaller of the two pivots, and
$Q_n$ be the (random) rank of the larger of the two. 
For a random permutation, $P_n$ and $Q_n$ have a joint distribution uniform over
all possible choices of a distinct pair of numbers from $\{1, \ldots, n\}$. 
That is to say,
\begin{align*}
		\Prob(P_n = p, Q_n = q) 
	&\wwrel= 
		\frac1{\tbinom n2}, 
	\qquad \mbox{for\ } 1 \le p < q \le n.
\end{align*}
It then follows that $P_n$ and $Q_n$ have the marginal distributions
\begin{align*}
		\Prob(P_n = p) 
	&\wwrel= \frac {n-p} {\tbinom n2}\,,
	 	\qquad \mbox {for\ }  p  = 1,\ldots, n-1 \;;
\\
		\Prob(Q_n = q) 
	&\wwrel= \frac {q-1} {\tbinom n2}\,, 
		\qquad \mbox {for\ }  q= 2, \ldots, n\;.
\end{align*}
Let $U_{(1)}$, and $U_{(2)}$ be the order
statistics of $U_1$ and $U_2$, 
two independent continuous $\uniform(0,1)$
random variables, i.\,e.\  
$$
	U_{(1)} \wrel= \min\{U_1, U_2\} \,, 
	\qquad \mbox{and}\qquad 
	U_{(2)} \wrel= \max\{U_1, U_2\} \;.
$$
The two order statistics have the joint density
\begin{align}
\label{eq:joint-density-u1u2}
		f_{(U_{(1)}, U_{(2)})} (x, y) 
	&\wwrel= \begin{cases}
		2, 	& \mbox{for\ } 0 < x < y < 1; \\
		0, 	& \mbox{elsewhere,}
	\end{cases}
\end{align}
and consequently have the marginal densities
\begin{align}
\label{eq:marginal-densities-u1-u2}
		f_{U_{(1)}} \! (x) 
	&\wwrel= \begin{cases}
		2(1-x), 	& \text{for } 0 < x  < 1; \\
		0, 		& \text{elsewhere,}
	\end{cases}
	\quad\text{and}\quad 
		f_{U_{(2)}} \! (y) 
	\wwrel= \begin{cases}
		2y,	& \text{for } 0 < y  < 1; \\
		0, 	& \text{elsewhere.}
	\end{cases} 
\end{align}
% Let $V$ be yet another independent 
% (of $U_1$ and $U_2$) random variable.
% 

The ensuing technical work requires all the variables to be defined
on the same probability space. 
Let $(U_1,U_2, V)$ be three independent $\uniform(0,1)$ random variables
defined on the probability space $([0,1], \borel_{[0,1]}, \lambda)$,
where $\borel_{[0,1]} = B \cap [0,1]$, for $B$ the usual
\weakemph{Borel sigma field} on the real line, and $\lambda$ is the
\weakemph{Lebesgue measure}.
We have
\begin{align}
\label{eq:distribution-Rn}
		R_n 
	&\wwrel\eqlaw 
		\lceil \.nV \rceil 
	\wwrel\eqlaw 
		\uniform[1 \range n] \;.
\end{align}

\section{Analysis of Yaroslavskiy's Dual-Partitioning}
\label{Sec:Yaroslavskiy}
An analysis of Quickselect using Yaroslavskiy's algorithm
(Algorithm~\ref{Alg:Yaroslavskiy}) requires a careful examination of this
algorithm. 
In addition, being a novel partitioning method, it is a goal
to analyze the method in its own right. 

Let $T_n$ be the number of comparisons exercised in the first call to
Yaroslavskiy's algorithm.
This quantity serves as a \emph{toll function} for the recurrence relation
underlying Quickselect:
For unfolding the recurrence at size $n$, we have to ``pay'' a toll of~$T_n$
comparisons.
The distribution of $T_n$  below
is given implicitly in the arguments of \cite{Wild2012} and
later used explicitly in \cite{wild2012thesis,Wild2013Quicksort}.

\begin{lemma}
\label{Lm:dist-Tn}
	The number of comparisons $T_n$ of Yaroslavskiy's partitioning method satisfies
	the following distributional equation conditional on $(P_n,Q_n)$:
	\begin{align*}
	T_n &\wwrel\eqlaw  n-1 \bin+ \Hyper(n-2, n - P_n - 1, Q_n-2)
\\*	    &\wwrel{\phantom{\eqlaw}} \phantom{n-1} \bin+ 
						\Hyper(n-2, \like{n-P_n-1,{}}{Q_n-2,{}} n-Q_n) 
\\*	    &\wwrel{\phantom{\eqlaw}} \phantom{n-1} \bin+ 
						3 \cdot \indicator_{\{A[Q_n] > \max\{A[1],A[n]\}\}}
	    \, .
	\end{align*}
\end{lemma}

\begin{theopargself}
\begin{corollary}[\cite{Wild2012}]
\label{Cor:expectation-Tn}
	The expectation of $T_n$ is given by
$$\E[T_n] \wwrel= \tfrac {19} {12} (n+1) - 3.$$
\end{corollary}
\end{theopargself}

\begin{theopargself}
\begin{corollary}[\cite{Wild2013Quicksort}]
\label{Cor:Tstar}
	The normalized number of comparisons $T_n^* \ce T_n \mathbin/ n$ converges to
	a limit $T^*$ in $L_2$: 
	$$
		T_n^* \wwrel{\convL2} T^* \,,
		\qquad \text{with} \qquad 
		T^* \wwrel\eqlaw 
			1 + U_{(2)} \bigl(2-  U_{(1)} - U_{(2)}\bigr) \;.
	$$
\end{corollary}
\end{theopargself}

\begin{remark}
We re-obtain the leading term coefficient of $\E[T_n]$ as the mean of the limit
distribution of $T_n^*$,
\begin{align*}
		%\E\bigl[1 + \bigr(U_{(2)} - U_{(1)}\bigr) \bigl(1- U_{(2)}\bigr) \bigr]
		\E[T^*] 
	&\wwrel= 1 + 	\int_0^1 \!\!\! \int_0^y y (2-x-y) 
						f_{(U_{(1)}, U_{(2)})} \!(x,y)\:
					dx \: dy 
	 \wwrel= \tfrac {19} {12} \;.
\end{align*}
The bivariate density $ f_{(U_{(1)}, U_{(2)})} (x,y)$ is given in
equation~\eqref{eq:joint-density-u1u2}.
\end{remark}

\section{Analysis of Yaroslavskiy's Algorithm for Random Rank}
\label{Sec:grand}

Let $\ui{C_n}{r}$ be the \emph{number of comparisons} made by Quickselect
under Yaroslavskiy's algorithm
on a random input of size~$n$ to seek the $r$th order statistic. 
While this variable is easy to analyze for extremal
values $r$ (nearly smallest and nearly largest),
it is harder to analyze for intermediate values of $r$, 
such as when $r = \lfloor 0.17 n \rfloor$. 
%Typically, the analysis for 
%median is hardest (when $r = \lfloor \frac 1 2 n \rfloor$). 
%It is also
%algorithmically hardest to find, see \cite{Schonhage1976}
% for complexity bounds. 

For a library implementation of Quickselect, though, it is quite natural to
consider $r$ to be part of the input (so that the only fixed parameter is $n$).
Analyzing $\ui{C_n}{R}$ when $R$ itself is random provides smoothing over
all possible values of $\ui{C_n}r$. 
We let $R=R_n$ be a random variable distributed like
$\uniform[1 \range n]$, i.\,e., every possible rank is requested with the same
probability. 
This rank randomization averages over the easy and hard
cases, which makes the problem of moderate complexity and amenable to analysis. 

In this case we can use the simplified notation $C_n \ce \ui{C_n}{R_n}$.
One seeks a grand average of all averages, a grand variance of all variances 
and a grand (average) distribution of all distributions 
in the specific cases of $r$ as a global measure over all possible
order statistics. This smoothing technique
was introduced in \cite{Mahmoud1995},
and was used successfully in \cite{Lent1996,Prodinger1995}. 
Panholzer and Prodinger give a generating function
formulation for grand averaging \cite{Panholzer1998}.

Right after the first round of partitioning, 
the two pivots (now moved to positions~$P_n$ and $Q_n$) 
split the data array into three subarrays:
$A[1\range P_n-1]$ containing keys with ranks smaller than $P_n$, 
$A[P_n+1\range Q_n-1]$ containing keys with ranks between $P_n$ 
and~$Q_n$ (both inclusively),
and $A[Q_n+1\range n]$
containing keys with ranks that are at least as large as $Q_n$.
Quickselect is then invoked recursively on one of the three
subarrays, depending on the desired order statistic.

As we have pairwise distinct elements almost surely, ranks are in one-to-one
correspondence with key values and the three subarrays contain ranks
\emph{strictly} smaller, between and larger than the pivots.
Therefore, we have the stochastic recurrence
\begin{equation}
		C_n
	\wwrel\eqlaw 
 		  		T_n 
 		\wbin+ C^{\phantom{n}}_{P_n-1}  \one {\{R_n < P_n\}} 
        \wbin+ C'_{Q_n- P_n-1}  \one {\{P_n < R_n < Q_n\}}  
        \wbin+ C^{\pprime}_{n -Q_n}  \one {\{R_n > Q_n\}} ,
\label{Eq:Cn}
\end{equation}
where, for each $i \ge 0$,
$C'_i \eqlaw C^{\pprime}_i \eqlaw C_i$, 
and $(C^{\phantom{n}}_{\smash{P_n}}, C'_{\smash{Q_n-P_n-1}}, C^{\pprime}_{\smash{n - Q_n}})$ 
are conditionally independent (in the sense that, given $P_n=p$, 
and $Q_n = q$,
$C^{\phantom{n}}_{\smash{p-1}}, C'_{\smash{q-p-1}}$,
and $C^{\pprime}_{n-q}$ are independent).

\subsection{Exact Grand Average}
\label{Subsec:grandave}

The distributional equation~\eqref{Eq:Cn} yields a recurrence for the average:
\begin{equation}
	\E[C_n] \wwrel= \E[T_n] + 3\, \E\bigl[C_{P_n-1}  \one {\{R_n < P_n\}}\bigr],
\label{Eq:AveCn}
\end{equation}
where symmetry is used to triple the term containing the first indicator.
By conditioning on $(P_n, Q_n)$ and the independent $R_n$, 
using Corollary~\ref{Cor:expectation-Tn} we get
\begin{align}
		\E[C_n] 
	&\wwrel= 
		\E[T_n] \wbin+ 3 \mkern-10mu \sum_{1\le p < q \le n\;} \sum_{r=1}^n
			\E\bigl[C_{P_n-1}  \one{\{R_n < P_n\}} \given P_n = p, Q_n =q, R_n = r\bigr]    
	\nonumber \\* &\wwrel\ppe\qquad\qquad\qquad{}
        	\times \Prob(P_n = p, Q_n = q, R_n = r)   
\nonumber \\ &\wwrel=
		\tfrac {19} {12} (n+1) - 3 \wwbin+ 3 \sum_{p=1}^n \sum_{r=1}^n
			\E\bigl[C_{p-1} \one {\{r < p\}}\bigr]    
	%\nonumber \\* &\wwrel\ppe\qquad\qquad{}
			\times \Prob(P_n = p) \, \Prob (R_n = r)  
\nonumber \\ &\wwrel=
		\tfrac {19} {12} (n+1) - 3 
		\wwbin+ \frac 6 {n^2(n-1)} \sum_{p=1}^n (p-1)(n-p) \, \E[C_{p-1}].
        \label{Eq:ECnrecurrence}
\end{align} 
This recurrence equation can be solved via generating functions.
Let 
$$A(z) \wrel\ce \sum_{n=0}^\infty n \,\E[C_n] \, z^n.$$
be the (ordinary) generating function for $n \, \E[C_n]$.

First, organize the recurrence~\eqref{Eq:ECnrecurrence} in the form
$$
  n^2(n-1)\, \E[C_n] \wwrel=   
                         n^2 (n-1) \bigl(\tfrac{19}{12} (n+1) - 3 \bigr)
                       \wbin+ 6 \sum_{p=1}^n (p-1)(n-p)\, \E[C_{p-1}]. 
$$
Next, multiply both sides of the equation by $z^n$ (for $|z| < 1$), and sum over $n\ge 3$,
the range of validity of the recurrence, to get
$$
	z^2 \sum_{n=3}^\infty \bigl(n^2(n-1) \, \E[C_n]\bigr) z^{n-2}
           \wwrel= 6 \sum_{n=3}^\infty \sum_{k=1}^n (n-k)  
						(k-1)\, \E[C_{k-1}] z^n
                    \wbin+ g(z)
\,,$$
where
\begin{align*}
	g(z) &\wrel= 
		\sum_{n=3}^\infty n^2 (n-1)\bigl(\tfrac {19} {12} (n+1) -3 \bigr)\, z^n 
\\ &
	\wrel=
		\frac {z^3} {(1-z)^5} \bigl( 7z^4 -35z^3 + 70 z^2 -64 z+ 60 \bigr) \;.
\end{align*}
Shifting summation indices and using the boundary conditions $C_0 = C_1 =0$,
and $C_2 = 1$, we extend the series to start at $n = 0$, and get
$$
		z^2\bigl(A^{\pprime}(z) - 2^2(2-1)\cdot 1 z^0\bigr) 
	\wwrel= 
			6 \sum_{n=0}^\infty \sum_{k=0}^{n} (n-k) z^{n-k} 
				\times \bigl(k\, \E[C_k]\bigr) z^k
			\wbin+ g(z).
$$ 
Finally, we get an Euler differential equation
$$z^2 A^{\pprime}(z) \wwrel= 6\frac {z^2} {(1-z)^2} A(z) \bin+ 4 z^2 \bin+ g(z),$$
to be solved under the boundary conditions
$A(0) = 0$, and $A'(0) = 0$.
The solution to this differential equation is
\begin{align*}
 		A(z)
 	&\wwrel=	
 		\frac 1 {300(1-z)^3} \Bigl( 
				2220 z - 510 z^2 + 830 z^3 - 1185z^4 
				+ 699 z^5 - 154 z^6
\\*	&	\qquad\qquad\qquad {} 
			+ 2220 (1-z) \ln(1-z)\Bigr). 
\end{align*}
Extracting coefficients of $z^n$, we find for $n \ge 4$,
\begin{align*}
		\E[C_n] 
	&\wwrel=  
		\tfrac {19} 6 n  - \tfrac {37} 5 H_n  
             +\tfrac {1183}  {100} -\tfrac {37}{5n} H_n 
             - \tfrac {71} {300n}
     \wwrel\sim  
     	\tfrac {19} 6 n, \qquad \mbox{as\ } n \to \infty.
\end{align*}
Proposition~\ref{Prop:grandave} is proved.
\qed

\subsection{Limit Distribution}
\label{Subsec:granddistribution}
Higher moments are harder to compute by direct recurrence
as was done for the mean. For instance, exact variance computation
involves rather complicated dependencies. We need a shortcut
to determine the asymptotic distribution (i.\,e.\ all asymptotic moments),
without resorting to exact calculation of each moment. 
A tool suitable for this task is the \weakemph{contraction method}.

The contraction method was introduced by Rösler~\cite{roesler1991limit} in the
analysis of the Quicksort algorithm, and it soon became a popular method
because of the transparency it provides in the
limit. Rachev and
Rüschendorf added several useful extensions~\cite{Rachev1995} and 
general contraction theorems,
and multivariate extensions are
available~\cite{neininger2001multivariate,Neininger2004,rosler2001analysis}. 
A valuable survey of the method was given by
Rösler~\cite{roesler2001contraction}. 
Neininger gives a variety of applications to random combinatorial structures and
algorithms such as random search trees, random recursive trees, random digital
trees and Mergesort~\cite{Neininger2004}. 
The contraction method has also been used in the context of classic
Quickselect~\cite{Gruebel1996,rosler2004quickselect}. 
Other methods for the analysis of Quickselect have been used;
for example, Grübel uses Markov chains~\cite{Gruebel1998}.

We shall use the contraction method to find the grand distribution
of Quickselect's number of comparisons under rank smoothing.

\smallskip
By dividing~\eqref{Eq:Cn} by $n$ and rewriting the
fractions, we find
\begin{align}
		\frac{C_n} {n} 
	&\wwrel\eqlaw 
		  \frac{C_{P_n-1}}{P_n-1} \cdot \frac {P_n-1} {n} 
		  	\indicator_ {\{R_n < P_n\}}
	\notag\\*&\wwrel\ppe\quad {}	
  		+ \frac {C'_{Q_n-P_n-1}} {Q_n-P_n-1}  \cdot \frac{Q_n-P_n-1} {n}
			\indicator_ {\{P_n < R_n < Q_n\}}
	\notag\\*&\wwrel\ppe\quad {}	
	  	+ \frac{C^{\pprime}_{n-Q_n}}{n - Q_n} \cdot \frac {n-Q_n} n 
			\indicator_ {\{R_n > Q_n\}}
	\notag\\*&\wwrel\ppe\quad {}			 
		\wbin+ \frac{T_n} {n} \;.
\label{eq:Cn/n-recurrence}
\end{align}
This equation is conveniently expressed in terms of the \emph{normalized} random
variables $C_n^* \ce C_n / n$, that is
\begin{align*}
		C_n^* 
	&\wwrel\eqlaw 
	  		C_{P_n-1}^* \frac {P_n-1} {n} \indicator_ {\{R_n < P_n\}}  
	\wbin+ 	\Csp_{Q_n-P_n-1} \frac{Q_n-P_n-1} {n} \indicator_ {\{P_n < R_n <
	Q_n\}} \\*&\wwrel{\phantom{\eqlaw}} \qquad {} 
	\wbin+	\Cspp_{n-Q_n} \frac{n-Q_n} {n} \indicator_ {\{R_n > Q_n\}} 
	\wbin+	T_n^*,
\end{align*}
where for each $j\ge 0$, $\Csp_j \rel\eqlaw \Cspp_j \rel\eqlaw C^*_j$ 
and each of the families $\{C_j^*\}$, $\{\Csp_j\}$, $\{\Cspp_j\}$,
$\{T_j\}$, and $\{R_j\}$ is comprised of totally independent random variables.
This representation suggests a 
limiting functional equation as follows. 
If $C_n^*$ converges to a limit $C^*$, so will $C_{P_n-1}^*$ because $P_n\to
\infty$ almost surely, and it is plausible to guess that the combination 
\smash{$ C_{P_n-1}^*  \tfrac{P_n-1}{n} \indicator_ {\{R_n <  P_n\}}$}
converges in distribution to $ C^* U_{(1)} \indicator_ {\{V < U_{(1)}\}}$. 
Likewise, it is plausible to guess that 
\begin{align*}
		\Csp_{Q_n -P_n-1} \frac {Q_n -P_n-1} {n} \indicator_{\{P_n < R_n < Q_n\}}
	&\wwrel\convD
		\Csp \bigl(U_{(2)}-U_{(1)}\bigr)
	 		\indicator_{\{ U_{(1)} < V <  U_{(2)} \}} \,, 
\intertext{and}	 		
	 	\Cspp_{n-Q_n} \frac {n-Q_n} {n} \indicator_ {\{R_n > Q_n\}} 
	 &\wwrel\convD
	 	\Cspp \bigl(1-U_{(2)}\bigr) 
	 		\indicator_ { \{V >  U_{(2)}\} } \,,
\end{align*}
where $\Csp \rel\eqlaw \Cspp \rel\eqlaw C^*$, and $(C^*, \Csp , \Cspp)$ are
totally independent.

To summarize, if $C_n^*$ converges in distribution to a limiting random variable $C^*$,
one can guess that the limit satisfies the following distributional
equation:
\begin{align}
		C^*  
	&\wwrel\eqlaw 
		  		U_{(1)} \indicator_{\{V < {U_{(1)}}\}} \, C^*
		\wbin+ 	\bigl(U_{(2)} - U_{(1)}\bigr) \indicator_ {\{U_{(1)} < V < U_{(2)}\}}
		\,
				\Csp
	\nonumber \\* &\wwrel{\phantom{\eqlaw}} \qquad {} 
		\wbin+	\bigl(1-U_{(2)}\bigr) \indicator_ {\{V > U_{(2)}\}}  \, \Cspp
		\wwbin+ T^*,
	\label{eq:Cstar-direct-limit-equation-copy}	
\end{align}
with $\bigl(C^*, \Csp, \Cspp, (U_{(1)}, U_{(2)}), V\bigr)$  being totally
independent, and $T^*$ as given in Corollary~\ref{Cor:Tstar}.

The formal proof of convergence is done by coupling the random variables
$C_n$, $P_n$, $Q_n$, and $R_n$ to be defined
on the same probability space, and showing that the \emph{distance} between
the distributions of $C^*_n$ and $C^*$ converges to~0 in some suitable metric
space of distribution functions. 
Here, we use the \weakemph{Zolotarev metric} $\zeta_2$, for which Neininger and
Rüschendorf give convenient contraction theorems \cite{Neininger2004}.
The technical details of using these theorems are provided in
Appendix~\ref{app:proof-convergence-grand-avg}. 
Finally, convergence in $\zeta_2$ implies the claimed convergence in
distribution and in second moments.

\medskip
The representation in~\eqref{eq:Cstar-direct-limit-equation-copy} admits direct
calculation of asymptotic mean and variance.
Taking expectations on both sides and exploiting symmetries gives
\begin{align*}
		\E[C^*] 
	&\wwrel=	\E\bigl[U_{(1)} \indicator_ {\{V < U_{(1)}\}} \bigr] \,\E[C^*] 
		\wbin+ 	\E\bigl[\bigl(U_{(2)} - U_{(1)}\bigr) 
						\indicator_ {\{U_{(1)} < V < U_{(2)}\}}  \bigr]\, \E[\Csp] 
\\*	&\wwrel\ppe\qquad {}  
		\wbin+ \E\bigl[\bigl(1-U_{(1)}\bigr) \indicator_ {\{V > U_{(2)}\}} \bigr]\,
					\E[\Cspp] 
		\wwbin+ \E[T^*]    
\\	&\wwrel= 	3 \, \E\bigl[U_{(1)} \indicator_ {\{V < U_{(1)}\}} \bigr] \,\E[C^*] 
		\wbin+ \tfrac {19} {12} \;.
\end{align*}
So, we compute
\begin{align*}
		\E\bigl[U_{(1)} \indicator_ {\{V < U_{(1)}\}} \bigr]  
	&\wwrel= 	\int_{x=0}^1 \int_{v=0}^1 x \: \indicator_ {\{v < x\}}
						\,f_{U_{(1)}}\!(x) \:f_V (v) \: dv \, dx 
\\ 	&\wwrel=	2\int_{x=0}^1 \int_{v=0}^x x(1-x)\, dv \, dx 
 	 \wwrel= 	\tfrac 1 6 \;.
\end{align*}
It follows that
$$
	\E[C^*] \wwrel= \tfrac {19} 6\,,
	\qquad\text{and, as } n \to\infty\,, \qquad 
	\E[C_n ] \wwrel\sim  \tfrac {19} 6 \. n \;.
$$

Similarly, we can get the asymptotic variance
of $C_n$. We only sketch this calculation.
First, square the distributional
equation~\eqref{eq:Cstar-direct-limit-equation-copy}, then take expectations. 
There will appear ten terms on the right-hand side. 
The three terms involving $(C^*)^2$ are symmetrical, and the three
terms involving cross-products of indicators are 0 (the indicators
are for mutually exclusive events). 
By independence, we have
\begin{align*}
		\E\bigl[(C^*)^2\bigr] 
	&\wwrel= 	3\, \E\bigl[ U_{(1)}^2 \indicator_ {\{V < U_{(1)\}}} \bigr] \,
   					\E\bigl[ (C^*)^2 \bigr] 
\\*	&\wwrel\ppe\qquad {}
		\wbin+	2\, \E\bigl[ T^* \, U_{(1)} 
						\indicator_ {\{V < U_{(1)}\}} \bigr] \, \E[C^*]
\\*	&\wwrel\ppe\qquad {}
    	\wbin+	2\, \E\bigl[ T^*\bigl(U_{(2)} - U_{(1)}\bigr) 
             			\indicator_ {\{U_{(1)} < V < U_{(2)}\}}  \bigr] \, \E[\Csp]
\\*	&\wwrel\ppe\qquad {}
		\wbin+	2\, \E\bigl[ T^*\bigl(1-U_{(2)}\bigr) 
						\indicator_ {\{V>U_{(2)}\}} \bigr]\, 
					\E[\Cspp] 
\\*	&\wwrel\ppe\qquad {}
		\wbin+	\E\bigl[ (T^*)^2 \bigr]. 
\end{align*}
We show the computation for one of these ingredients:
$$
	\E\bigl[U_{(1)}^2 \indicator_ {\{V < U_{(1)\}}} \bigr] 
	\wwrel=  2\int_{y=0}^1\int_{x=0}^y\int_{v=0}^1 
					x^2 \indicator_ {\{v < x\}}  
				\, dv\, dx\, dy 
	\wwrel= \tfrac 1 {10} \;.
$$
After carrying out similar calculations and using Corollary~\ref{Cor:Tstar}, we
obtain
\begin{align*}
		\E\bigl[(C^*)^2\bigr] 
   &\wwrel= 
   		  \tfrac 3 {10} \, \E\bigl[(C^*)^2\big] 
   		+ 2\Bigl(
   			\tfrac {43} {180}\, \E[C^*]
	   		+ \tfrac {53} {180}\, \E[\Csp] 
	   		+ \tfrac {1} {4}\, \, \E[\Cspp]
   		\Bigr)  
   		+ \tfrac {229} {90} \;.
\end{align*}
We can solve for $\E\bigl[(C^*)^2\bigr]$ and 
by inserting $\E[C^*] = \E[\Csp] = \E[\Cspp] = \tfrac{19}6$ 
get
\begin{align*}
		\E\bigl[(C^*)^2\bigr] 
   &\wwrel= \tfrac{10}7 \Bigl( 
   					2\bigl(
   					  \tfrac {43} {180} 
   					+ \tfrac {53} {180} 
       				+ \tfrac {1} {4}\bigr) \tfrac{19}6  
       			+ \tfrac {229} {90}
       		\Bigr) 
    \wwrel= \tfrac{193}{18}	
    \;.
\end{align*}
The variance follows:
\begin{align*}
		\V[C_n] 
	&\wwrel=		\E[C_n^2] - \bigl(\E[C_n]\bigr)^2 	
\\	&\wwrel\sim  	\Bigl(\E[(C^*)^2] - \bigl(\E[C^*]\bigr)^2\Bigr) n^2
\\	&\wwrel=		\bigl(\tfrac{193}{18} - \tfrac{361}{36}\bigr) n^2 
	 \wwrel=			\tfrac{25}{36} n^2
	 \,, \qquad \text{as } n \to \infty \;.
\end{align*}

\medskip
We next present an explicit (unique) solution to the distributional
equation~\eqref{eq:Cstar-direct-limit-equation-copy}. 
A~random variable with this distribution takes the form of a perpetuity.

\begin{lemma}
\label{Lm:equivalantdistribution}
	Every solution $C^*$ of \eqref{eq:Cstar-direct-limit-equation-copy}
	also satisfies the distributional equation \eqref{Eq:Cstar}.
\end{lemma}

The proof is by showing that the characteristic functions for the solutions
of both equations coincide.
Detailed computations are given
Appendix~\ref{app:proof-lemma-equivalantdistribution}. 

\smallskip
The representation in Lemma~\ref{Lm:equivalantdistribution}
allows us
to obtain an expression for~$C^*$ as a sum of products of independent random variables.
Toward this end, let $X_1, X_2, \ldots$ be independent copies of $X^*$, and
let $Y_1, Y_2, \ldots$ be independent copies of $g(X^*, W^*)$,
then
$$
	C^*	
		\wwrel\eqlaw	Y_1 + X_1 C^*  
		\wwrel\eqlaw	Y_1 + X_1 (Y_2 + X_2  C^* ).
$$
Note that because $C^*$ is independent of both $X_1$ and $Y_1$, the $X$ 
and $Y$ introduced in the iteration must be independent copies of $X_1$ and $Y_1$.
Continuing the iterations (always introducing new independent random variables), 
we arrive at
\begin{align}
  C^* 
		&\wwrel\eqlaw 
			Y_1 + X_1 Y_2 + X_1 X_2 (Y_3 + X_3 \, C^*) \nonumber\\
		&\wwrel{\like{\eqlaw}{\vdots}}	\nonumber\\[-1ex]
		&\wwrel\eqlaw 
			\sum_{j=1}^M \Bigl(Y_j \prod_{k=1}^{j-1} X_k \Bigr)
					+ X_1 X_2 \cdots X_M \, C^*,
	\label{Eq:strong}
\end{align}
for any positive integer $M$. However, by the strong law of large numbers,
$$
		\frac 1 M \ln (X_1 X_2 \ldots X_M)
	\wwrel\almostsure 
		\E[\ln X^*] 
	\wwrel= -\tfrac 5 6, 
	\qquad\mbox{as\ } {M\to \infty} \;,
$$
and
$$
	X_1 X_2 \ldots X_M \wwrel\almostsure 0, 
		\qquad\mbox{as\ } {M\to \infty} \;.
$$
Hence, we can proceed with the limit of~\eqref{Eq:strong} and write
$$
		C^*  
	\wwrel\eqlaw 
		\sum_{j=1}^\infty \Bigl(Y_j \prod_{k=1}^{j-1} X_k \Bigr)
	\;.
$$

\section{Analysis of Yaroslavskiy's Algorithm for Extremal Ranks}
\label{Sec:extremal}

The methods used for deriving results for dual-pivot Quickselect
to locate a key of random rank carry over to the case of a relatively 
small or relatively large extremal order statistic.
We shall only sketch the arguments and results for the case $r=1$, as they
closely mimic what has been done for a random rank.

Let $\hat C_n \ce \ui{C_n}{1}$ be the number
of key comparisons required by Quickselect running with Yaroslavskiy's algorithm
to find the smallest element 
in an array $A[1\range n]$
of random data.  
If the smaller of the two pivots (of random rank $P_n$)
in the first round is not the smallest key, the algorithm  
always pursues the leftmost subarray $A[1\range P_n -1]$. 
We thus have the recurrence
\begin{align}
		\hat C_n 
	&\wwrel\eqlaw 
		\hat C_{P_n-1} + T_n
	\;.
	\label{Eq:Cnhat}
\end{align}

\subsection{Exact Mean}
\label{Subsec:extremalave}

Equation~\eqref{Eq:Cnhat} yields a recurrence for the average
$$\E[\hat C_n] \wwrel= \E[T_n] +  \E\bigl[\hat C_{P_n-1} \bigr].$$
Conditioning on $P_n$, we find
\begin{align}
		\E[C_n]	
	&\wwrel=	
		\E[T_n] \bin+ \sum_{p=1}^n 
				\E\bigl[C_{P_n-1} \given P_n = p\bigr]  
				\, \Prob(P_n = p) 
\nonumber \\
	&\wwrel=
			\tfrac {19} {12} (n+1) - 3 
				\wbin+ \frac1{\binom n2} \sum_{p=1}^n (n-p)\, \E[C_{p-1}]. 
	\label{Eq:ECnextremalrecurrence}
\end{align} 
This recurrence equation can be solved via generating functions, 
by steps very similar to what we did in Subsection~\ref{Subsec:grandave},
and we only give an outline of intermediate steps.
If we let 
$$\hat A(z) \wwrel\ce \sum_{n=0}^\infty \, \E[\hat C_n] z^n \,,$$
multiply~\eqref{Eq:ECnextremalrecurrence} by $n(n-1)z^n$ and sum over $n\ge 3$, 
we get an Euler differential equation
$$
		z^2 \bigl(\hat A''(z)-2 \bigr) 
	\wwrel= 
		\frac {2 z^2 \hat A(z)} {(1-z)^2} + h(z)
\,,$$
with
\begin{align*}
	h(z) &\wwrel\ce \sum_{n=3}^\infty \bigl(\tfrac{19}{12}(n+1)-3\bigr) n(n-1) z^n
	\wwrel=
		\frac{z^3}{2(1-z)^4}
			\bigl(
				-7z^3 + 28z^2 - 42z + 40
			\bigr)
		\;.
\end{align*}
This differential equation is
to be solved under the boundary conditions
$\hat A(0) = 0$, and $\hat A'(0) = 0$.
The solution  is
\begin{align*}
		\hat A(z)
	&\wwrel=	
		\frac 1 {24(1-z)^2} \Bigl( 36 z + 42 z^2 -28 z^3 + 7 z^4 
					+ 12  \bigl(3 -6 z^2 + 4z^3 - z^4\bigr) \ln(1-z)\Bigr) \;. 
\end{align*}
Extracting coefficients of $z^n$, we find, for $n \ge 4$,
\begin{align*} 
		\E\bigl[\hat C_n\bigr]
	&\wwrel=	
			\frac 1 {24n(n-1)(n-2)} \Bigl(57n^4 - 48n^3 H_n - 178 n^3 
				+ 144 n^2 H_n 
\\*	&	\qquad\qquad\qquad{} 
				+ 135 n^2 - 96nH_n - 14 n + 24
			\Bigr)
\\	&\wwrel\sim 
		\tfrac {19} 8 n, \qquad \mbox{as\ } n \to \infty \;.
\end{align*}

\subsection{Limit Distribution}
\label{Subsec:extremaldistribution}

By similar arguments as in case of random ranks 
we see that $\hat C^*_n \ce \hat C_n / n$ approaches $\hat
C^*$, a random variable satisfying the distributions equation
\begin{equation}
		\hat C^* 
	\wwrel\eqlaw 
		U_{(1)} \, \hat C^* \bin+ T^* \,,
\label{Eq:Cstarextremal}
\end{equation}
where \smash{$(U_{(1)}, T^*)$} is independent of $\hat C^*$.
We formally establish convergence in law and second moments by showing that
the distance of the distributions of $\hat C_n^*$ and $\hat C^*$ diminishes
to~$0$ in the Zolotarev metric $\zeta_2$.
We go through the technical work in
Appendix~\ref{app:proof-convergence-extremal}, 
using a handy theorem of Neininger and
Rüschendorf \cite{Neininger2004} , which standardized the approach.

\medskip
The representation in equation~\eqref{Eq:Cstarextremal} allows us
to obtain an expression for~$\hat C^*$ by an unwinding process
like that we used for $C^*$; one gets
$$
		\hat C^*  
	\wwrel\eqlaw 
		\sum_{j=1}^\infty \Bigl(Y_j \prod_{k=1}^{j-1} X_k \Bigr)
	\,,
$$
with $\{X_j\}_{j=1}^\infty$ and $\{Y_j\}_{j=1}^\infty$ being two families of
totally independent random variables whose members are all distributed
like~$U_{(1)}$ respectively $T^*$.
This completes a sketch of the proof of Theorem~\ref{Theo:extremal}.
\qed

\section{Conclusion}
\label{Sec:remarks}

In this paper, we discussed the prospect of running Quickselect  
making use of a dual-pivot partitioning strategy by Yaroslavskiy, which recently
provided a speedup for Quicksort and is used today in the library sort of
Oracle's Java~7. 
It has been proven that, for sorting, the total number of
comparisons becomes smaller on average, upon
using Yaroslavskiy's algorithm compared to the classic 
single-pivot variant~\cite{Wild2012}. 
Even if a single partitioning phase may need more comparisons than in the
classic case, the reduced sizes of the subproblems to be processed
recursively\,---\,the input to be sorted is partitioned into three instead of
two parts\,---\,lead to an overall saving.

The speedup in Quicksort by Yaroslavskiy's partitioning algorithm 
raises the hope for similar improvements in Quickselect. 
However, our detailed analysis, presented in this paper and summarized in
Table~\ref{tab:results}, proves the opposite:
When searching for a (uniformly) random rank, we find an expected number of
comparisons of asymptotically $\tfrac{19}6 n = 3.1\overline6 n$ 
for a Quickselect variant running under Yaroslavskiy's algorithm, 
as opposed to $3n$ for classic Quickselect. 
For extremal cases, i.\,e., for ranks close to the minimum or maximum, an
asymptotic average of $\tfrac{19}8 n = 2.375 n$ comparisons is needed,
whereas the classic algorithm only uses $2n$. 
Though not considered here in detail, similar trends are observed for the
number of swaps: 
In expectation, Yaroslavskiy's algorithm,
also needs more swaps than the classic variant.

\begin{table}
	\begin{minipage}{\linewidth}
	\small
	\renewcommand{\thempfootnote}{\alph{mpfootnote}}
	\def\an#1{$#1 n$}
	\def\stdev{std.\,dev.}
	\def\expec{expectation}
	\mbox{}\hfill
	\begin{threeparttable}
	\begin{tabular}{r@{\quad}l@{\;\;}lcc}
		\toprule
			\multicolumn{2}{c}{\multirow2*{\bfseries Cost Measure}}
			&& \bfseries Quickselect with & \bfseries Classic Quickselect \\
			&& \;\;error
			 & \bfseries Yaroslavskiy's Algorithm & \bfseries with Hoare's Algorithm
			  \\
		\midrule
		&&&\multicolumn2c{\textsl{when selecting a uniformly chosen order statistic}}
		\\
		\cmidrule{4-5}
		\multirow{2}*{Comparisons}
			& \expec & $\Oh(\log n)$ 	
				& \an{3.1\overline6}
				& \an3 \tnotex{tn:mahmoud95} \\
			& \stdev & $o(n)$	
				& \an{0.8\overline3}
				& \an1 \tnotex{tn:mahmoud95}  \\[1ex]
		\multirow{1}*{Swaps}
			& expectation & $\Oh(\log n)$
				& \an1 \tnotex{tn:yaros-swaps}
				& \an{0.5} \tnotex{tn:classic-swaps}\\[1ex]
		\midrule
		&&&\multicolumn2c{\textsl{when selecting an extremal order statistic}} \\
		\cmidrule{4-5}
		\multirow{3}*{Comparisons}
			& \expec & $\Oh(\log n)$ 	
				& \an{2.375}
				& \an2 \tnotex{tn:mahmoud95} \\[.25ex]
			& \multirow2*{\stdev} & \multirow2*{$o(n)$}	
				& \llap{small: }\an{0.512551}
				& \multirow2*{\an{0.707107} \tnotex{tn:mahmoud95}}  \\
			&&& \llap{large: }\an{0.598087}	  \\[1.5ex]
		\multirow{1}*{Swaps}
			& expectation & $\Oh(\log n)$
				& \an{0.75} \tnotex{tn:yaros-swaps}
				& \an{0.\overline3} \tnotex{tn:classic-swaps}\\[.25ex]
		\bottomrule
	\end{tabular}
	\begin{footnotesize}
	\begin{tablenotes}
		\item[\dag] \label{tn:mahmoud95} 
			see \cite{Mahmoud1995}; Theorems~1 and~2.
		\item[\ddag] \label{tn:yaros-swaps} 
			by the linear relation of expected swaps and comparisons, we
			can insert the toll function for swaps from~\cite{Wild2013Quicksort}.
		\item[*] \label{tn:classic-swaps} 
			see \cite{Hwang2000}; integrating over $\alpha\in[0,1]$
			resp.\ taking $\alpha\to0$ in eq~(12).
	\end{tablenotes}
	\end{footnotesize}
	\end{threeparttable}
	\hfill\mbox{}
	\end{minipage}
	\caption{
		Main results of this paper and the corresponding results for classic
		Quickselect from the literature.
	}
	\label{tab:results}
\end{table}

The observed inferiority of dual-pivoting in Quickselect goes beyond the case of 
Yaroslavskiy's algorithm. 
Aum\"uller and Dietzfelbinger~show that \emph{any} dual-pivoting method must use at
least~$\tfrac32 n$ comparisons on average for a single partitioning step
\cite[Theorem~1]{Aumuller2013}.
Even with such an optimal method, Quickselect needs $3n + o(n)$ comparisons on
average to select a random order statistic and $2.25n + o(n)$ comparisons when
searching for extremal elements; 
no improvement over classic Quickselect in both measures.

\smallskip
Our analysis provides deeper insight than just the average case\,---\,we derived
variances and fixed-point equations for the distribution of the number of
comparisons.
Even though of less practical interest than the average considerations, it is
worth noting that the variance of the number of comparisons made by
Quickselect under Yaroslavskiy's algorithm is significantly smaller than for
the classic one, making actual costs more predictable.

\smallskip
We can give some intuition based on our analysis on why dual pivoting does not
improve Quickselect. 
As already pointed out, the new algorithm may need more
comparisons for a single partitioning round
than the classic strategy, but leads to smaller subproblems to be handled
recursively. 
In classic Quickselect, using pivot~$p$ and searching for a random rank $r$, we
use $n$ comparisons to exclude from recursive calls either $n-p$ elements, if
$r<p$, or $p$ elements, if~$r>p$, or all $n$ elements, if $r=p$. 
Averaging over all $p$ and $r$, this implies that a single comparison helps to
exclude $\frac{1}{3} + o(1)$ elements from further computations, as $n\to\infty$. 
When interpreting our findings accordingly, Quickselect under Yaroslavskiy's
algorithm excludes asymptotically only $\frac{6}{19} \approx 0.3125$
elements per comparison on average. 
We have to conclude that the reduction in subproblem sizes is not sufficient
to compensate for the higher partitioning costs.

Nevertheless, the attempts to improve the Quickselect algorithm are not a
complete failure. 
Preliminary experimental studies give some hope for a faster algorithm in
connection with cleverly chosen pivots (comparable to the median-of-three strategy
well-known in the classic context), especially for presorted data.
Future research may focus on this scenario, trying to identify an optimal
choice for the pivots.
Related results are known for classic
Quickselect~\cite{Martinez2010,Martinez2001} and Yaroslavskiy's algorithm in
Quicksorting~\cite{Wild2013Alenex}.

Furthermore, it would be interesting to extend our analysis to the number of
bit comparisons instead of atomic key comparisons. 
This is especially of interest in connection with nonprimitive 
data types like
strings. 
However, in this context one typically has to deal with much more complicated
analysis for the resulting subproblems no longer preserve randomness
in the subarrays
(see~\cite{Fill2009} for corresponding results for classic Quickselect).
As a consequence, the methods used in this paper are no longer applicable.
 
On the grounds of current knowledge, though\,---\,i.\,e., the comparison model that
we studied in this paper\,---\,we can recommend to practitioners 
the use of Yaroslavskiy's partitioning in Quicksort, but not in Quickselect.

\begin{small}
\bibliography{quicksort-refs}   % name your BibTeX data base
\end{small}

\clearpage
\appendix
\section*{{\LARGE Appendix}}
\small

\numberwithin{theorem}{section}

\section{Spacings and Subproblem Sizes}
\label{app:definition-spacings}

The only data aspect that plays any role in a comparison-based
sorting or selection method is the \emph{relative ranking}. For example,
comparing  80 to 60 is the same as comparing 2 to 1 (and the same
action is taken in both cases, like a swapping, for instance). All
probability models that give a random permutation of ranks almost
surely are therefore equivalent from the point of view of  
a comparison-based sorting method. 
For example, sorting data from \emph{any} continuous distribution gives rise to
the same distribution of complexity measures, such as the number of comparisons
or swaps. 
We might as well work through the analysis using one convenient continuous data
model, like the $\uniform(0,1)$. 
In this appendix, we assume our data to come from such a uniform density,
as was done in~\cite{MahmoudPittel1988,Wild2013Quicksort}, which gives a
convenient definition of random permutations of increasing sample sizes (and
infinite random permutations, as well).

We use a notation that symmetrizes the reference to the three subarrays 
(see~\cite{Wild2013Quicksort}). Instead of referring to the sizes
of the three subarrays by $P_n-1$, $Q_n-P_n-1$, and $n-Q_n$,
we simply call them $\ui{I_j}n$, for $j=1,2,3$.
Moreover, $P_n$ and $Q_n$ are an inconvenient basis for the transition to the
limit as they take values from the discrete set $\{1,\ldots,n\}$; 
hence we use an alternative description of the distribution of
$\ui{\vect I}n = (\ui{I_1}n,\ui{I_2}n,\ui{I_3}n)$:

Denote by $\vect S = (S_1,S_2,S_3)$ 
the \emph{spacings} induced by the two independent random variables $U_1$ and
$U_2$ distributed uniformly in $(0,1)$; formally we have
\begin{align*}
	(S_1,S_2,S_3)  &\wwrel=
	( U_{(1)} ,\, U_{(2)}-U_{(1)} ,\, 1-U_{(2)} )\,, %\\
\end{align*}
for
\begin{align*}
	 U_{(1)} \ce \min\{U_1,U_2\},
	 	 \qquad  \text{ and } \qquad U_{(2)} \ce \max\{U_1,U_2\}
	 	\;.
\end{align*}
It is well-known that $\vect S$ is uniformly
distributed in the standard $2$-simplex \cite[p.\,133f]{David2003},
i.\,e.\ $(S_1,S_2)$ has density
\begin{align*}
	f_S(x_1,x_2) &\wwrel=
		\begin{cases}
			2, & \text{for } x_1,x_2 \ge 0 \rel\wedge x_1+x_2 \le 1, \\
			0, & \text{otherwise}\;.
		\end{cases}
\end{align*}
$S_3 = 1-S_1-S_2$ is fully determined by $(S_1,S_2)$.

We can express the distribution of the sizes of the three
subarrays 
$\ui{\vect I}{n} = (\ui{I_1}{n},\ui{I_2}{n},\ui{I_3}{n})$
generated in the first partitioning round based on $\vect S$: 
We have 
$\ui{I_1}{n} + \ui{I_2}{n} + \ui{I_3}{n} = n-2$, and conditional on $\vect S$,
the vector $\ui{\vect I}{n}$ has a multinomial distribution:
\begin{align*}
	\ui{\vect I} n &\wwrel\eqdist \multinomial(n-2;\,S_1,S_2,S_3).
\end{align*}

\begin{lemma}
\label{lem:asymptotic-I}
	We have	for $j\in\{1,2,3\}$ the convergence, as $n\to\infty$,
	\begin{align*}
		\frac{\ui{I_j}n}{n}
			&\wwrel{\convL 2}  S_j
		\;.
	\end{align*}
\end{lemma}
\begin{proof}
By the \weakemph{strong law of large numbers}, $\ui{I_j}n / n$ converges to
$S_j$ almost surely.
Moreover, $|\ui{I_j}n / n|$ is bounded by $1$, and the statement follows.
\end{proof}

\clearpage
\section{Proof of Convergence: Rank Smoothing}
\label{app:proof-convergence-grand-avg}

We are going to apply the following theorem by Neininger and
Rüschendorf~\cite{Neininger2004}, which we restate for the reader's
convenience:
%to show convergence of (suitably normalized) random variables in the Zolotarev
%metric~$\zeta_s$:

\begin{theopargself}
\begin{theorem}[{\cite[Theorem~4.1]{Neininger2004}}]
	\label{thm:neininger-thm-4.1}
	Let $(X_n)_{n\in\N}$ be a sequence of $s$-integrable, $0<s\le 3$, random
	variables satisfying the recurrence
	\begin{align*}
			X_n
		&\wwrel\eqlaw
			\sum_{r=1}^K \ui{A_r}n \ui{X_{\ui{I_r}n}}r \wrel+ \ui bn \,,
			\qquad n \ge n_1 \,,
	\end{align*}
	where $K\in\N$ is a constant, 
	%$(\ui{A_1}n,\ldots,\ui{A_K}n, \ui bn, \ui In)$
	$(\ui{\vect A\!}n, \ui bn, \ui{\vect I}n)$
	and $(\ui{X_n}1)_{n\in\N}$, \dots, $(\ui{X_n}K)_{n\in\N}$ are independent and
	$\ui{X_i}r$ is distributed like $X_i$ for all $r=1,\ldots,K$ and $i\ge0$.
	Assume further for all $n$: 
	(a)~if $0<s\le1$ that $X_n$ has finite variance;
	(b)~if $1<s\le2$ additionally that $\E[X_n] = 0$ and 
	(c)~if $2<s\le 3$ additionally that $\V[X_n]=1$.
	Moreover, assume the following conditions:
	\begin{enumerate}[itemsep=0ex,leftmargin=3em,label=(\Alph*)]
		\item \rule{0pt}{10pt} \label{cond:convergence-of-coeffs}
			\smash{$(\ui{A_1}n,\ldots,\ui{A_K}n,\ui bn) 
				\wrel{\convL s} 
				(A_1^*,\ldots,A_K^*,b^*)$};
		\item \rule{0pt}{10pt} \label{cond:contraction-of-coeffs}
			\smash{$\E \bigl[
					\sum_{r=1}^K | A_r^* |^s
				\bigr] < 1$};
		\item \rule{0pt}{10pt} \label{cond:non-degeneracy-of-coeffs}
			For all $c\in\N$ and $r=1,\ldots,K$ holds 
				\smash{$\E\bigr[
					\indicator_{\{\ui{I_r}n\le c \rel\vee \ui{I_r}n = n\}}
					| \ui{A_r}n |^s
				\bigr] \wrel\to 0$}\,.
	\end{enumerate}
	Then \smash{$\lim\limits_{n\to\infty} \zeta_s(X_n,X) = 0$} for a random
	variable $X$ whose distribution is given as the unique fixed point of
	\begin{align*}
		X &\wwrel\eqlaw \sum_{r=1}^K A_r^* \ui Xr \rel+ b^* 
	\end{align*}
	among all distributions with 
	(a)~finite $s$th moments, if $0<s\le1$,
	(b)~finite $s$th moments and mean zero, if $1<s\le2$, and 
	(c)~finite $s$th moments, variance one and mean zero, if $2<s\le 3$.
	Here, $(A_1^*,\ldots,A_K^*,b^*)$ and $\ui X1$, \dots, $\ui XK$ are independent
	and $\ui Xr$ is distributed like $X$ for $r=1,\ldots,K$.
	\qed
\end{theorem}
\end{theopargself}

See Section~2 of \cite{Neininger2004} for definition and discussion of the
Zolotarev metrics $\zeta_s$. 
For the application in this paper, it suffices to restate the following
properties where we write \smash{$X_n \convZ s X$} to mean
$\zeta_s(X_n,X)
\to 0$:
%\lim_{n\to\infty} \zeta_s(X_n,X) \wwrel= 0
\vspace{-1ex} 
\begin{enumerate}[itemsep=-.5ex,leftmargin=3em,label=(\Roman*)]
\item \label{prop:zeta-implies-convD}
	%Convergence in $\zeta_s$ implies convergence in law
	\rule{0pt}{12pt}
	\smash{$X_n \rel{\convZ s} X \wrel\implies X_n \rel\convD X$}, see
	\cite[p.\,382]{Neininger2004}.
\item \label{prop:zeta1-implies-conv-mean}
	%Convergence in $\zeta_1$ implies convergence of means
	\rule{0pt}{12pt}
	\smash{$X_n \rel{\convZ1} X \wrel\implies \E[X_n] \rel\to \E[X]$}, see
	\cite[Remark p.\,398]{Neininger2004}.
\item \label{prop:zeta2-implies-conv-var}
	%Convergence in $\zeta_2$ implies convergence of variance
	\rule{0pt}{12pt}
	\smash{$X_n \rel{\convZ2} X \wrel\implies \V[X_n] \rel\to \V[X]$}, see
	\cite[Remark p.\,398]{Neininger2004}.
\end{enumerate}
\smallskip

To prove convergence in distribution and in  second moments of the number of
comparisons used by Yaroslavskiy's algorithm under rank smoothing, we apply
Theorem~\ref{thm:neininger-thm-4.1} with $K=3$ and $s=2$. 
For $s=2$, Theorem~\ref{thm:neininger-thm-4.1} requires \emph{centered} random
variables\,---\,so we cannot directly use $C_n^*$. 
Guessing that $\V[C_n] = \Theta(n^2)$, we define
\begin{align*}
	C^{\.\circ}_n &\wwrel\ce \frac{C_n - \E[C_n]}n \;.
\end{align*} 
(Note the difference between $C^*_n$ and $C^{\.\circ}_n$.)
Whence, $\E[C_n]$ is given in Proposition~\ref{Prop:grandave}.
Equation~\eqref{eq:Cn/n-recurrence} %on page~\pageref{eq:Cn/n-recurrence} 
can be written in terms of $C^{\.\circ}_n$ by subtracting $\E[C_n]$ on both
sides and rewriting the right hand side as follows:
\begin{align*}
	C^{\.\circ}_n &\wwrel\eqlaw 
		%\sum_{j=1}^3 C^{\.\circ}_{I_j} \cdot \frac{I_j}n \indicator_{\mathcal E_j^{(n)}}
		  C^{\.\circ}_{I_1}  
		  		\frac{\ui{I_1}n}n \indicator_{\ui{\mathcal E_1}n}
		\bin+ C^{\.\circ\.\prime}_{I_2}  
				\frac{\ui{I_2}n}n \indicator_{\ui{\mathcal E_2}n}
		\bin+ C^{\.\circ\.\prime\prime}_{I_3}  
				\frac{\ui{I_3}n}n \indicator_{\ui{\mathcal	E_3}n}
%	\\*&\wwrel\pe\qquad{}
		\wbin+ \frac1n \Bigl(
						T_n \bin- \E[C_n] 
					\bin+ \sum_{j=1}^3
						\indicator_{\ui{\mathcal E_j}n} 
						\E\bigl[ C_{I_j} \given \ui{I_j}n \bigr]
				\Bigr)\;. 
\end{align*}
Here $C_n^{\.\circ\.\prime}$ and $C_n^{\.\circ\.\prime\prime}$ are independent
copies of $C_n^{\.\circ}$ and $\ui{\mathcal E_j}n$ is the event that
the search continues in subproblem~$j$ (for $j=1,2,3$):
\begin{align*}
	\ui{\mathcal E_1}n \wrel\ce \{R_n<P_n\}			\,,\qquad
	\ui{\mathcal E_2}n \wrel\ce \{P_n < R_n < Q_n\}	\,, \qquad
	\ui{\mathcal E_3}n \wrel\ce \{Q_n < R_n\}		\;.
\end{align*}
In the proof that follows,
we express these events via the distributional
equation~\eqref{eq:distribution-Rn}
\begin{align*}
	\ui{\mathcal E_1}n \wrel\eqlaw \{V < \tfrac{\ui{I_1}n}n\}			\,,\qquad
	\ui{\mathcal E_2}n \wrel\eqlaw \{\tfrac{\ui{I_1}n+1}n < V < \tfrac{I_1+I_2+1}n\}	\,,
	\qquad 
	\ui{\mathcal E_3}n \wrel\eqlaw \{\tfrac{\ui{I_1}n + \ui{I_2}n + 2}n < V\}		\;.
\end{align*}

\smallskip
Now, we show that the three
conditions~\ref{cond:convergence-of-coeffs}, \ref{cond:contraction-of-coeffs} 
and~\ref{cond:non-degeneracy-of-coeffs} of
Theorem~\ref{thm:neininger-thm-4.1} are fulfilled:

\begin{description}[font=\bfseries,itemsep=1ex]

\item[Cond.\,\ref{cond:convergence-of-coeffs}]

We first consider the coefficients $\ui{A_j}n$.
For $j=1,2,3$, we have 
\begin{align*}
		\ui{A_j}n \wrel= \tfrac{\ui{I_j}n}n \indicator_{\ui{\mathcal E_j}n} 
	&\wwrel{\convL2} 
		S_j \indicator_{\mathcal E_j} \,,
\end{align*} 

where the limiting events $\mathcal E_j$ are defined as
\begin{align*}
	\mathcal E_1 \wrel\ce \{V<S_1\}			\,,\qquad
	\mathcal E_2 \wrel\ce \{S_1 < V < S_1+S_2\}	\,, \qquad
	\mathcal E_3 \wrel\ce \{S_1+S_2 < V\}		\;.
\end{align*}
It is essential that $\ui{\mathcal E_j}n$ and $\mathcal E_j$ are defined
in terms of the \emph{same} random variable $V$; this couples the events and allows
us to show the above convergence. 
The proof is a standard, but somewhat tedious computation:

For $j=1$, we condition on $\vect I$ and $\vect S$ to expand the indicator
variables:
\begin{align*}
	\E\Bigl[\Bigl(
		  \tfrac{\ui{I_1}n}n \indicator_{\mathcal E_1^{(n)}} 
		- S_1 \indicator_{\mathcal E_1}
		\Bigr)^2\Bigr]
	&\wwrel\le
		\E\biggl[
			\Prob\Bigl(V <\min\bigl\{\tfrac{\ui{I_1}n}n,S_1\bigr\} \Bigr) \cdot
			\Bigl(\tfrac{\ui{I_1}n} n - S_1\Bigr)^2 
	\\*	&\qquad     {} 
			+ \Prob\Bigl(
				\min\bigl\{\tfrac{I_1^{(n)}}n,S_1\bigr\} 
				< V <
				\max\bigl\{\tfrac{I_1^{(n)}}n,S_1\bigr\}
			\Bigr) \cdot 
			\max\nolimits^2\bigl\{\tfrac{\ui{I_1}n} n,S_1 \bigr\} 
		\biggr]
	\\&\wwrel\le
		\E\Bigl[
			\Bigl(\tfrac{\ui{I_1}n}n - S_1\Bigr)^2 
	                +\Bigl|\tfrac{\ui{I_1}n}n - S_1\Bigr|
	    \,\Bigr]
\\	&\wwrel=
		\Bigl\| \tfrac{I_1^{(n)}}n - S_1 \Bigr\|_2
		\wbin + 
		\Bigl\| \tfrac{I_1^{(n)}}n - S_1 \Bigr\|_1
\\&\wwrel\to 
		0\,, \qquad\text{for } n \to \infty\,,
\end{align*}
where the second inequality uses that the factors are bounded by $1$ uniformly
in $n$ and the last step follows by Lemma~\ref{lem:asymptotic-I}.

The convergence of the terms for $j=3$ are very much the same by symmetry. 
In the case for $j=2$, we have some more cases to distinguish, as the
corresponding events $\ui{\mathcal E_2}n$ and $\mathcal E_2$ contain upper
\emph{and} lower bounds. 
We skip details similar in nature to the case $j=1$.

\smallskip
The second part of condition~\ref{cond:convergence-of-coeffs} concerns the
convergence of the toll function. We show:
\begin{align}
		\frac1n \Bigl(
						T_n \bin- \E[C_n] 
					\bin+ \sum_{j=1}^3
						\indicator_{\ui{\mathcal E_j}n} 
						\E\bigl[ C_{\ui{I_j}n} \given \ui{I_j}n \bigr]
				\Bigr)
	&\wwrel{\convL2}
		T^* \bin+ \tfrac{19}6 \Bigl(-1 + 
			\sum_{j=1}^3 S_j \indicator_{\mathcal E_j}
		\Bigr) \;.
\end{align}
As \smash{$X_n \convL p X$ and $Y_n \convL p Y$ implies $X_n + Y_n \convL p
X+Y$}, we can show convergence of each of the summand individually.
We established \smash{$T_n / n \to T^*$} in Corollary~\ref{Cor:Tstar} and by
Proposition~\ref{Prop:grandave}, we have $\E[C_n] / n \to 19/6$. 
For the remaining sum, consider the first summand as an example;
the others are similar.
Using Proposition~\ref{Prop:grandave} once more and the independence of $V$ and
$(\vect S, \vect I)$, we find
\begin{align*}
		\E \biggl[ \Bigl(
			\frac1n \indicator_{\ui{\mathcal E_1}n} \E\bigl[  
						C_{I_1} \given \ui{I_1}n \bigr]
			\wbin- \tfrac{19}6 S_1 \indicator_{\mathcal E_1}
		\Bigr)^2 \biggr]
	&\wwrel \le 
		\E \biggl[ \Bigl(
			\indicator_{\ui{\mathcal E_1}n}
				\frac{\frac{19}6 \ui{I_1}n + o(n)}n
			\wbin- \tfrac{19}6 S_1 \indicator_{\mathcal E_1}
		\Bigr)^2 \biggr]
\\	&\wwrel \le
		\bigl( \tfrac{19}6 \bigr)^2\:
		\Bigl\| 
			\tfrac{\ui{I_1}n}n \indicator_{\ui{\mathcal E_1}n} 
			- S_1 \indicator_{\mathcal E_1}
		\Bigr\|_2
		\wbin+ o(1)
\\&\wwrel\to 
		0\,, \qquad\text{for } n \to \infty\,,
\end{align*}
where the last step follows from the $L_2$-convergence shown above for the
coefficients.

\item[Cond.\,\ref{cond:contraction-of-coeffs}]
We have to show that
\vspace{-2ex}
\begin{align*}
	\E\Bigl[ \sum_{j=1}^3 
		|S_j \indicator_{\mathcal E_j}|^2 
	\Bigr] 
	&\wwrel < 1 \;.
\end{align*}
By linearity of the expectation, it suffices to consider the summands in
isolation. We compute
\begin{align*}
		\E\bigl[ \bigl(S_1 \indicator_{\mathcal E_1}\bigr)^2 \bigr]
	&\wwrel=
		\int_{s_1=0}^1 \int_{s_2=0}^{1-s_1} \int_{v=0}^1 
			s_1^2 \indicator_{\{v<s_1\}} \cdot 2  
		\;  d v \: d s_2 \:  d s_1 
\\	&\wwrel=
		\int_{s_1=0}^1 (1-s_1) \int_{v=0}^{s_1} 
			2 s_1^2   
		\;  d v \:  d s_1
%\\	&
\wwrel=
		2 \int_{x=0}^1 (1-x) x^3 \; d x
	\wwrel= \tfrac1{10} \;.
\end{align*}
The other summands are symmetric, so we find 
$\E\bigl[ \sum_{j=1}^3 
		|S_j \indicator_{\mathcal E_j}|^2 
	\bigr] = \tfrac3{10} < 1$.

\item[Cond.\,\ref{cond:non-degeneracy-of-coeffs}]

The third condition of Theorem~\ref{thm:neininger-thm-4.1} requires for $j =
1,2,3$ and any constant~$c \in\N$ that 
\begin{align*}
	\E\biggl[
		\indicator_{\{\ui{I_j}n \le \, c \wrel\vee \ui{I_j}n = \,n\}}
		\Bigl| \tfrac{\ui{I_j}n} n \indicator_{\ui{\mathcal E_j}n} \Bigr |^2
	\biggr] \wwrel\to 0
	\,, \qquad \text{as } n \to \infty \;. 
\end{align*}
As we have $\ui{I_j}n \le n-2$ by definition and 
\smash{$\bigl| \tfrac{\ui{I_j}n}n \indicator_{\ui{\mathcal E_j}n} \bigr|^2$}
is bounded from above by $1$ uniformly in $n$, it suffices to show
\begin{align*}
	\Prob\bigl( \ui{I_j}n \le c \bigr) \wwrel\to 0 \;. 
\end{align*}
But this directly follows from the weak law of large numbers as the expected
values grow linearly with~$n$, that is, 
$\E[\ui{I_j}n \given \vect S] = n S_j$.

\end{description}

\medskip\noindent
All conditions of Theorem~\ref{thm:neininger-thm-4.1} are fulfilled, so we
obtain $C^{\.\circ}_n \convZ2 C^{\.\circ}$ and the distribution of $C^{\.\circ}$
is the unique fixed point of the distributional equation
\begin{align}
\label{eq:C5star-limit-dist-equation}
	C^{\.\circ} &\wwrel\eqlaw
	       S_1 \indicator_{\mathcal E_1} C^{{\.\circ}}
	 \bin+ S_2 \indicator_{\mathcal E_2} C^{{\.\circ\.\prime}}
	 \bin+ S_3 \indicator_{\mathcal E_3} C^{{\.\circ\.\prime\prime}}
	\wwbin+ T^* - \tfrac{19}6
	\bin+ \tfrac{19}6 \sum_{j=1}^3 S_j\indicator_{\mathcal E_j}
\end{align}
among all centered distributions with finite second moments;
$C^{{\.\circ}\.\prime}$ and $C^{\.\circ\.\prime\prime}$ are independent
copies of~$C^{\.\circ}$.

\smallskip
It remains to transfer convergence of $C_n^{\.\circ}$ to $C_n^*$.
By Proposition~\ref{Prop:grandave}, we have
\begin{align*}
	C_n^* &\wwrel= C_n^{\.\circ} + \tfrac{19}6 + o(1) \;.
\end{align*} 
So, the asymptotic difference between the two is a (deterministic) constant and
$C_n^*$ converges to the limit $C^* = C^{\.\circ} + 19/6$. 
Inserting this into the limit equation~\eqref{eq:C5star-limit-dist-equation}
for~$C^{\.\circ}$, we obtain~\eqref{eq:Cstar-direct-limit-equation} and the
first part of Theorem~\ref{Theo:randomrank} is proved.
The second part of Theorem~\ref{Theo:randomrank} directly follows from
Lemma~\ref{Lm:equivalantdistribution}, which we prove in 
Appendix~\ref{app:proof-lemma-equivalantdistribution}.

\clearpage
\section{Proof of Convergence: Extreme Ranks}
\label{app:proof-convergence-extremal}

For extreme ranks, we can apply Theorem~5.1 of~\cite{Neininger2004}, restated
here for convenience:

\begin{theopargself}
\begin{theorem}[{\cite[Theorem~5.1]{Neininger2004}}]
	\label{thm:neininger-thm-5.1}
	Let $(Y_n)_{n\in\N}$ be a sequence of $s$-integrable, $0<s\le 3$, random
	variables satisfying the recurrence
	\begin{align*}
		Y_n &\wwrel\eqlaw \sum_{r=1}^K 
			\ui{Y_{\!\rule{0pt}{1.75ex}\ui{I_r}n}}r
		\wrel+ b_n \,, \qquad n\ge n_0\,,
	\end{align*}
	where $K\in\N$ is a constant, $(b_n, \ui In)$
	and $(\ui{Y_n}1)_{n\in\N}$, \dots, $(\ui{Y_n}K)_{n\in\N}$ are independent and
	$\ui{Y_i}r$ is distributed like $Y_i$ for all $r=1,\ldots,K$ and $i\ge0$.
	Additionally, we require $\Prob[\ui{I_r}n = n] \to 0$ as $n\to\infty$ and
	$\V[Y_n] > 0$ for $n\ge n_0$. 
	Assume further that there are functions $f,g:\N_0 \to \R_{\ge0}$ such that 
	(a)~if $0<s\le1$ we have $g(n)>0$ for large $n$;
	(b)~if $1<s\le2$ additionally $\E[Y_n] = f(n) +
	o\bigl(\sqrt{g(n)}\bigr)$ holds and 
	(c)~if $2<s\le 3$ additionally $\V[Y_n]=g(n) + o\bigl(g(n)\bigr)$ is satisfied.
	Moreover, assume the following conditions:
	\vspace{.5ex}
	\begin{enumerate}[itemsep=0ex,leftmargin=3em,label=(\Alph*)]
		\item \rule{0pt}{10pt} \label{cond:stabilization-condition-g}
			For all $r=1,\ldots,K$ holds
			\smash{$ \sqrt{g(\ui{I_r}n) / g(n)}
				\wrel{\convL s} 
				A_r^*$};
		\item \rule{0pt}{10pt} \label{cond:stabilization-condition-f}
			$\frac1{\sqrt{g(n)}} \bigl( 
					b_n - f(n) + \sum_{r=1}^K f(\ui{I_r}n)
				\bigr) \wrel{\convL s} b^*$;
		\item \rule{0pt}{10pt} \label{cond:contraction-of-coeffs-5.1}
			\smash{$\E \bigl[
					\sum_{r=1}^K (A_r^*)^s
				\bigr] < 1$}\,.
	\end{enumerate}
	\vspace{1ex}
	Then for $(X_n)$ defined by 
	$ X_n \ce \bigl(Y_n - f(n)\bigr) / \!\sqrt{g(n)}$,
	we have  
	\smash{$\lim\limits_{n\to\infty} \zeta_s(X_n,X) = 0$} 
	for a random variable $X$ whose distribution is given as the unique fixed
	point of
	\begin{align*}
		X &\wwrel\eqlaw \sum_{r=1}^K A_r^* \ui Xr \rel+ b^* 
	\end{align*}
	among all distributions with 
	(a)~finite $s$th moments, if $0<s\le1$,
	(b)~finite $s$th moments and mean zero, if $1<s\le2$, and 
	(c)~finite $s$th moments, variance one and mean zero, if $2<s\le 3$.
	Here, $(A_1^*,\ldots,A_K^*,b^*)$ and $\ui X1$, \dots, $\ui XK$ are independent
	and $\ui Xr$ is distributed like $X$ for $r=1,\ldots,K$.
	\qed
\end{theorem}
\end{theopargself}

This theorem is a special case of theorem we used in
Appendix~\ref{app:proof-convergence-grand-avg} and applies to distributional
recurrences where the (non-normalized) costs $Y_n$ of subproblems directly
contribute, without a factor in front of them.
Our equation~\eqref{Eq:Cnhat} is of this form.

We apply Theorem~\ref{thm:neininger-thm-5.1} 
with $K=1$, $s=2$ and
functions $f(n) = \tfrac{19}8 n$ and $g(n) = n^2$.
By Proposition~\ref{Prop:extremalave} we have 
\smash{$\E[\hat C_n] = f(n) + o\bigl(\sqrt{g(n)}\bigr)$} as needed. 
Theorem~\ref{thm:neininger-thm-5.1} then states convergence of the
centered variables 
\smash{$\hat C_n^{\.\circ} = \frac1n(\hat C_n - \tfrac{19}8n)$},
given conditions~\ref{cond:stabilization-condition-g},
\ref{cond:stabilization-condition-f} and~\ref{cond:contraction-of-coeffs-5.1}
are fulfilled.
We check the conditions:

\begin{description}[font=\bfseries,topsep=1ex]
\item[Cond.\,\ref{cond:stabilization-condition-g}]
We have by Lemma~\ref{lem:asymptotic-I}:
\vspace{-2ex}
\begin{align*}
	\sqrt{ \frac{g(I_1^{(n)})}{g(n)} }
	&\wwrel= \frac{I_1^{(n)}}n 
	\wwrel{\convL2} S_1 \;.
\end{align*}

\item[Cond.\,\ref{cond:stabilization-condition-f}]
Using Proposition~\ref{Prop:extremalave} and Corollary~\ref{Cor:Tstar},
we compute:
\vspace{-1.5ex}
\begin{align*}
		g(n)^{-1/2}\bigl( T_n - f(n) + f(\ui{I_1}n) \bigr)
	&\wwrel=
		T_n^* - \tfrac{19}8 + \tfrac{19}8 \tfrac{\ui{I_1}n}n
	\wwrel{\convL2} T^* + \tfrac{19}8(S_1-1) \;. 
\end{align*}

\item[Cond.\,\ref{cond:contraction-of-coeffs-5.1}]
A simple computation shows
$\E\bigl[ S_1^2 \bigr] = \tfrac16 < 1$.

\end{description}

\smallskip\noindent
All requirements of Theorem~\ref{thm:neininger-thm-5.1} are fulfilled and we
conclude that
\begin{align}
\label{eq:Chat5star-definition}
		\hat C_n^{\.\circ}
		\wrel= \frac{\hat C_n - \tfrac{19}8 n}n
		\wrel= \hat C_n^* - \tfrac{19}8
	&\wwrel{\convZ2}
		\hat C^{\.\circ} \,, 
\end{align}
where the distribution of $\hat C^{\.\circ}$ is obtained as unique fixed point
of
\begin{align}
\label{eq:Chat5star-limit-dist-equation}
		\hat C^{\.\circ}
	&\wwrel\eqlaw
		S_1 C^{\.\circ} \wbin+ T^* + \tfrac{19}8 (S_1-1) \,,
\end{align}
among all centered distributions with finite second moments.
The constant difference between $\hat C_n^*$ and~$\hat C_n^{\.\circ}$
implies that $\hat C_n^*$ converges to $C^{\.\circ} + \tfrac{19}8$.
Inserting $C^* = C^{\.\circ} + \tfrac{19}8$
in~\eqref{eq:Chat5star-limit-dist-equation} we obtain
equation~\eqref{Eq:Chat}; Theorem~\ref{Theo:extremal} is proved.

\clearpage
\section{Proof of Lemma~\protect\ref{Lm:equivalantdistribution}}
\label{app:proof-lemma-equivalantdistribution}

To prove that the two distributional equations are equivalent, 
we show that the \weakemph{characteristic function} $\phi_{C^*}(t)$ of variable
$C^*$ fulfilling equation~\eqref{eq:Cstar-direct-limit-equation-copy} is also
the characteristic function of a solution
of~\eqref{Eq:Cstar}.
From~\eqref{eq:Cstar-direct-limit-equation-copy}, we find:
\begin{align*}
\phi_{C^*}(t) 
&\wwrel= 
\E\biggl[\exp\Bigl(it\Bigl( 
U_{(1)} \indicator_{\{V < {U_{(1)}}\}} \, C^*
\wbin+ 	\bigl(U_{(2)} - U_{(1)}\bigr) \indicator_ {\{U_{(1)} < V < U_{(2)}\}}
\\*&\wwrel\ppe\qquad {}
\bin+	\bigl(1-U_{(2)}\bigr) \indicator_ {\{V > U_{(2)}\}}  \, \Cspp
\wbin+ T^*\Bigr)\Bigr)\biggr] \\
&\wwrel= {} \mathbin{\phantom{+}}
\int_{y=0}^1 \int_{x=0}^y \int _{v=0} ^x 
\E\bigl[e^{it(xC^*+T^*)}\bigr] f_{(U_{(1)}, U_{(2)}, V)} (x,y,v) 
\, dv \, dx \, dy
\\*&\wwrel\ppe {} 
+  \int_{y=0}^1 \int_{x=0}^y \int _{v=x} ^y 
\E\bigl[e^{it((y-x)C^* +T^*)}\bigr] f_{(U_{(1)}, U_{(2)}, V)} (x,y,v) 
\, dv \, dx \, dy
\\*&\wwrel\ppe {}
+ \int_{y=0}^1 \int_{x= 0}^y \int _{v=y} ^1 
\E\bigl[e^{it((1-y)C^* +T^*)}\bigr] f_{(U_{(1)}, U_{(2)}, V)} (x,y,v) 
\, dv \, dx \, dy  \\[0.5ex]
&\wwrel= {} \mathbin{\phantom{+}}
2 \int_{y=0}^1 \int_{x= 0}^y  
x\, \E\bigl[e^{it(xC^* + 1 +y(2-x-y))}\bigr] 
\, dx \, dy
\\*&\wwrel\ppe {} 
+ 2\int_{y=0}^1 \int_{x= 0}^y 
(y-x) \, \E\bigl[e^{it((y-x)C^* + 1 +y(2-x-y))}\bigr] 
\, dx \, dy 
\\*&\wwrel\ppe {} 
+ 2 \int_{y=0}^1 \int_{x= 0}^y
(1-y)\, \E\bigl[e^{it((1-y)C^*+1 + y(2-x-y))}\bigr] 
\, dx \, dy \;.                  
\end{align*}
In the middle integral, make the change of variables
$$ x = 1-u, \qquad y = 1-u+v,$$
and in the rightmost integral make the change of variables
$$ x = 1-u, \qquad y = 1-v,$$
to get the characteristic function in the form
\begin{align*}
\phi_{C^*}(t) 
&\wwrel= {} \mathbin{\phantom{+}} 
2 \int_{u=0}^1 \int_{v=0}^u
\phi_{C^*}(tv) \, v e^{it(1+u(2-v-u))}
\, dv \, du 
\\*&\wwrel\ppe {} 
+ 2 \int_{u=0}^1 \int_{v=0}^u 
\phi_{C^*}(tv) \, v e^{it(1+(1-u+v) (2u-v))} 
\, dv \, du
\\*&\wwrel\ppe {}
+ 2 \int_{u=0}^1 \int_{v=0}^u 
\phi_{C^*}(tv) \, v e^{it(1+(1-v)(u+v))} 
\, dv \, du \\
&\wwrel=
\int_{w=0}^1  \int_{x=0}^w \Bigl(  
\tfrac 1 3 e^{it(1+w(2-x-w))} 
+ \tfrac 1 3 e^{it(1+(1+x-w)(2w-x))} 
+ \tfrac 1 3 e^{it(1+(1-w)(x+w))} 
\Bigr)
\\*&\wwrel\pe \qquad\qquad\qquad {} \times  (6x) \phi_{C^*} (tx)   \,
dx\, dw,
\end{align*}
which is also the characteristic function of $X^* C^* + g(X^*, W^*)$.
\qed

\end{document}

%% file: quickselect-paper-preamble-arXiv.tex
\usepackage{tocbibind} % Show References section
\addtokomafont{caption}{\sffamily\small}
\addtokomafont{captionlabel}{\sffamily\textbf}
\setcapmargin{2em}

\usepackage{microtype}

\usepackage{amsmath}
\usepackage{amsthm}

\usepackage{xspace}
\usepackage{enumitem}
\usepackage{booktabs}
\usepackage{multicol}
\usepackage{multirow}
\usepackage[referable]{threeparttablex}

\usepackage{clrscode3e}

\def\EndIf{\End\li\kw{end if} }
\def\EndWhile{\End\li\kw{end while} }
\newcommand\arrayA{\ensuremath{A}\xspace}

\usepackage{float}
\floatstyle{ruled}
\newfloat{algorithm}{tbp}{loa}
\floatname{algorithm}{Algorithm}

\usepackage{amsmath}
\usepackage{amssymb}
\usepackage{bm}
\usepackage{tikz}
\usepackage{colonequals}

\theoremstyle{definition}

\newtheorem{corollary}{Corollary}
\newtheorem{remark}{Remark}

\theoremstyle{plain}
\newtheorem{theorem}{Theorem}
\newtheorem{proposition}{Proposition}
\newtheorem{lemma}{Lemma}

\renewcommand\proof[1][]{%
    \ifthenelse{\equal{#1}{}}{%
        \textit{Proof: }%
    }{%
        \textit{Proof #1:}%
    }%
}

\newenvironment{theopargself}{}{}

\usepackage[hyphens]{url}
\usepackage{ifthen}
\usepackage{placeins}

\usepackage[
			final,
			unicode=true, 
            bookmarks=true,
            bookmarksnumbered=true,
            bookmarksopen=true,
            breaklinks=true,
            pdfborder={0 0 0}, % no border around links
            colorlinks=false   % do not make links appear red
]{hyperref}
\hypersetup{pdftitle={Analysis of Quickselect under Yaroslavskiy's Dual-Pivoting Algorithm}, 
 pdfauthor={Sebastian Wild, Markus Nebel and Hosam Mahmoud},
 pdfsubject={},
 pdfkeywords={Quicksort, Quickselect, combinatorial probability, alorithm, recurrence, average-case 
 analysis, grand average, perpetuity, fixed point, metric, contraction}%
}

\vfuzz \hfuzz	

\allowdisplaybreaks[3]

\newcommand\weakemph[1]{\textsl{#1}}
 
\newdimen\makeboxdimen
\newcommand\makeboxlike[3][l]{%
\setbox0=\hbox{#2}%
\global\makeboxdimen=\wd0%
\setbox1=\hbox{\makebox[\makeboxdimen][#1]{%
\makebox[0pt][#1]{#3}%
}}%
\ht1=\ht0%
\dp1=\dp0%
\box1%	
}

\newcommand\indicator{{\bf 1}}
\def\borel{\mathbb B}

\newcommand\uniform{\mathrm{Uniform}}
\newcommand\Hyper{{{\rm Hypergeo}}}
\newcommand\multinomial{\mathrm{Mult}}

\newcommand\range{\, .. \, }

\newcommand\one[1]{{\bf 1}_{#1}}
\newcommand\convD{%
	{\buildrel {\raisebox{-1ex}{$\scriptstyle\cal D$}} 
	\over \longrightarrow}%
}
\newcommand\convZ[1]{%
	\mathrel{\overset
	{\raisebox{-2ex}{$\scriptscriptstyle\zeta_{#1}$}}
	\longrightarrow }%
}
\newcommand\convL[1]{%
	\mathrel{\overset
	{\raisebox{-2ex}{$\scriptstyle L_{#1}$}}
	\longrightarrow }%
}

\newcommand\eqlaw{%\buildrel D\over= 
	\mathchoice{%display
		\mathrel{\overset{\raisebox{0ex}{$\scriptstyle \cal D$}}=}%
	}{%text
		\mathrel{\overset{\raisebox{-1ex}{$\scriptscriptstyle \cal D$}}=}%
	}{%script
		\mathrel{\overset{\cal D}=}%
	}{%scriptscript
		\mathrel{\overset{\cal D}=}%
	}%
}
\let\eqdist\eqlaw
\newcommand\given{\, \vert \, }

\newcommand\E{{\bf E}}
\newcommand\V{{\bf Var}}
\newcommand\almostsure{\buildrel a.s. \over  \longrightarrow}

\newcommand\Prob{{\bf Prob}}

\let\ce\colonequals
\def\.{\mskip1mu}

\newcommand\N{\mathbb N}
\newcommand\R{\mathbb R}

\newcommand\rel[1]{\mathrel{\:{#1}\:}}
\newcommand\wrel[1]{\mathrel{\;{#1}\;}}
\newcommand\wwrel[1]{\mathrel{\;\;{#1}\;\;}}
\newcommand\bin[1]{\mathbin{\:{#1}\:}}
\newcommand\wbin[1]{\mathbin{\;{#1}\;}}
\newcommand\wwbin[1]{\mathbin{\;\;{#1}\;\;}}

\newcommand\pe{\phantom{{}={}}}
\newcommand\ppe{\phantom=}

\newcommand\pprime{\prime\prime}
\newcommand\Csp{C^{*\.\prime}}
\newcommand\Cspp{C^{*\.\pprime}}

\newcommand\Oh{\mathcal O}

\newcommand\like[3][c]{\makeboxlike[#1]{\ensuremath{#2}}{\ensuremath{#3}}}

\newcommand{\ui}[2]{%
	\mathchoice{%display
		#1^{(#2)}%
	}{%text
		#1^{%
			\smash{\raisebox{.75pt}{\mbox{\tiny$($}}}%
			#2%
			\smash{\raisebox{.75pt}{\mbox{\tiny$)$}}}%
		}%
	}{%script
		#1^{\smash{(}\rule[-.75pt]{0pt}{2.75pt}#2\smash{)}}%
	}{%scriptscript
		#1^{(#2)}%
	}%
}

\newcommand{\vect}[1]{\boldsymbol{\mathbf{#1}}}

%% file: quickselect-paper-arXiv.bbl
\begin{thebibliography}{10}

\bibitem{Iksanov2009}
Gerold Alsmeyer, Alex Iksanov, and Uwe Rösler.
\newblock On distributional properties of perpetuities.
\newblock {\em Journal of Theoretical Probability}, 22:666--682, 2009.

\bibitem{Aumuller2013}
Martin Aum\"{u}ller and Martin Dietzfelbinger.
\newblock {Optimal Partitioning for Dual Pivot Quicksort}.
\newblock March 2013.

\bibitem{Bentley1984}
Jon Bentley.
\newblock {Programming pearls: how to sort}.
\newblock {\em Communications of the ACM}, 27(4):287--291, April 1984.

\bibitem{Chung2001}
Kai~Lai Chung.
\newblock {\em {A Course in Probability Theory}}.
\newblock Academic Press, 3rd edition, 2001.

\bibitem{David2003}
Herbert~A. David and Haikady~N. Nagaraja.
\newblock {\em {Order Statistics (Wiley Series in Probability and
  Statistics)}}.
\newblock Wiley-Interscience, 3rd edition, 2003.

\bibitem{Embrechts1997}
Paul Embrechts, Claudia Kl\"{u}ppelberg, and Thomas Mikosch.
\newblock {\em {Modelling Extremal Events}}.
\newblock Springer Verlag, Berlin, Heidelberg, 1997.

\bibitem{Fill2013}
James~Allen Fill.
\newblock {Distributional convergence for the number of symbol comparisons used
  by QuickSort}.
\newblock {\em The Annals of Applied Probability}, 23(3):1129--1147, June 2013.

\bibitem{Fill2009}
James~Allen Fill and Tak\'{e}hiko Nakama.
\newblock {Analysis of the Expected Number of Bit Comparisons Required by
  Quickselect}.
\newblock {\em Algorithmica}, 58(3):730--769, March 2009.

\bibitem{Gruebel1998}
Rudolf Gr\"{u}bel.
\newblock {Hoare's Selection Algorithm: A Markov Chain Approach}.
\newblock {\em Journal of Applied Probability}, 35(1):36--45, 1998.

\bibitem{Gruebel1996}
Rudolf Gr\"{u}bel and Uwe R\"{o}sler.
\newblock {Asymptotic Distribution Theory for Hoare's Selection Algorithm}.
\newblock {\em Advances in Applied Probability}, 28(1):252--269, 1996.

\bibitem{hennequin1989combinatorial}
Pascal Hennequin.
\newblock {Combinatorial analysis of Quicksort algorithm}.
\newblock {\em Informatique th\'{e}orique et applications}, 23(3):317--333,
  1989.

\bibitem{hennequin1991analyse}
Pascal Hennequin.
\newblock {\em {Analyse en moyenne d’algorithmes : tri rapide et arbres de
  recherche}}.
\newblock {PhD Thesis}, Ecole Politechnique, Palaiseau, 1991.

\bibitem{Hoare1961}
C.~A.~R. Hoare.
\newblock {Algorithm 65: Find}.
\newblock {\em Communications of the ACM}, 4(7):321--322, July 1961.

\bibitem{Hoare1962}
C.~A.~R. Hoare.
\newblock {Quicksort}.
\newblock {\em The Computer Journal}, 5(1):10--16, January 1962.

\bibitem{Hwang2000}
Hsien-Kuei Hwang and Tsung-Hsi Tsai.
\newblock {Quickselect and Dickman function}.
\newblock {\em Combinatorics, Probability and Computing}, 11, 2000.

\bibitem{Kirschenhofer1998}
Peter Kirschenhofer and Helmut Prodinger.
\newblock {Comparisons in Hoare's Find Algorithm}.
\newblock {\em Combinatorics, Probability and Computing}, 7(01):111--120, 1998.

\bibitem{Knuth1998}
Donald~E. Knuth.
\newblock {\em {The Art Of Computer Programming: Searching and Sorting}}.
\newblock Addison Wesley, 2nd edition, 1998.

\bibitem{Lent1996}
Janice Lent and Hosam~M. Mahmoud.
\newblock {Average-case analysis of multiple Quickselect: An algorithm for
  finding order statistics}.
\newblock {\em Statistics \& Probability Letters}, 28(4):299--310, 1996.

\bibitem{Mahmoud2010}
Hosam~M. Mahmoud.
\newblock {Distributional analysis of swaps in Quick Select}.
\newblock {\em Theoretical Computer Science}, 411:1763--1769, 2010.

\bibitem{Mahmoud1995}
Hosam~M. Mahmoud, Reza Modarres, and Robert~T. Smythe.
\newblock {Analysis of quickselect : an algorithm for order statistics}.
\newblock {\em Informatique th\'{e}orique et applications}, 29(4):255--276,
  1995.

\bibitem{MahmoudPittel1988}
Hosam~M. Mahmoud and Boris Pittel.
\newblock On the joint distribution of the insertion path length and the number
  of comparisons in search trees.
\newblock {\em Discrete Applied Mathematics}, 20(3):243--251, July 1988.

\bibitem{Martinez2010}
Conrado Mart\'{\i}nez, Daniel Panario, and Alfredo Viola.
\newblock {Adaptive sampling strategies for quickselects}.
\newblock {\em ACM Transactions on Algorithms}, 6(3):1--45, June 2010.

\bibitem{Martinez2009}
Conrado Mart\'{\i}nez and Helmut Prodinger.
\newblock {Moves and displacements of particular elements in Quicksort}.
\newblock {\em Theoretical Computer Science}, 410(21–23):2279--2284, 2009.

\bibitem{Martinez2001}
Conrado Mart\'{\i}nez and Salvador Roura.
\newblock {Optimal Sampling Strategies in Quicksort and Quickselect}.
\newblock {\em SIAM Journal on Computing}, 31(3):683, 2001.

\bibitem{neininger2001multivariate}
Ralph Neininger.
\newblock {On a multivariate contraction method for random recursive structures
  with applications to Quicksort}.
\newblock {\em Random Structures \& Algorithms}, 19(3-4):498--524, 2001.

\bibitem{Neininger2004}
Ralph Neininger and Ludger R\"{u}schendorf.
\newblock {A General Limit Theorem for Recursive Algorithms and Combinatorial
  Structures}.
\newblock {\em The Annals of Applied Probability}, 14(1):378--418, 2004.

\bibitem{Panholzer1998}
Alois Panholzer and Helmut Prodinger.
\newblock {A generating functions approach for the analysis of grand averages
  for multiple QUICKSELECT}.
\newblock {\em Random Structures \& Algorithms}, 13(3-4):189--209, 1998.

\bibitem{Prodinger1995}
Helmut Prodinger.
\newblock {Multiple Quickselect---Hoare's Find algorithm for several elements}.
\newblock {\em Information Processing Letters}, 56(3):123--129, 1995.

\bibitem{Rachev1995}
Svetlozar~T. Rachev and Ludger R\"{u}schendorf.
\newblock {Probability Metrics and Recursive Algorithms}.
\newblock {\em Advances in Applied Probability}, 27(3):770--799, 1995.

\bibitem{roesler1991limit}
Uwe R\"osler.
\newblock A limit theorem for ``quicksort''.
\newblock {\em Informatique th{\'e}orique et applications}, 25(1):85--100,
  1991.

\bibitem{rosler2001analysis}
Uwe R{\"o}sler.
\newblock On the analysis of stochastic divide and conquer algorithms.
\newblock {\em Algorithmica}, 29(1):238--261, 2001.

\bibitem{rosler2004quickselect}
Uwe R\"{o}sler.
\newblock {QUICKSELECT revisited}.
\newblock {\em Journal of the Iranian Statistical Institute}, 3:271--296, 2004.

\bibitem{roesler2001contraction}
Uwe R\"{o}sler and Ludger R\"{u}schendorf.
\newblock {The contraction method for recursive algorithms}.
\newblock {\em Algorithmica}, 29(1):3--33, 2001.

\bibitem{Sedgewick1975}
Robert Sedgewick.
\newblock {\em {Quicksort}}.
\newblock {PhD Thesis}, Stanford University, 1975.

\bibitem{Sedgewick1977}
Robert Sedgewick.
\newblock {The analysis of Quicksort programs}.
\newblock {\em Acta Informatica}, 7(4):327--355, 1977.

\bibitem{wild2012thesis}
Sebastian Wild.
\newblock {Java 7's Dual Pivot Quicksort}.
\newblock Master thesis, University of Kai\-sers\-lau\-tern, 2012.

\bibitem{Wild2012}
Sebastian Wild and Markus~E. Nebel.
\newblock {Average Case Analysis of Java 7's Dual Pivot Quicksort}.
\newblock In Leah Epstein and Paolo Ferragina, editors, {\em ESA 2012}, volume
  7501 of {\em LNCS}, pages 825--836. Springer, 2012.

\bibitem{Wild2013Quicksort}
Sebastian Wild, Markus~E. Nebel, and Ralph Neininger.
\newblock {Average Case and Distributional Analysis of Java 7's Dual Pivot
  Quicksort}.
\newblock {\em ACM Transactions on Algorithms}, accepted for publication.

\bibitem{Wild2013Alenex}
Sebastian Wild, Markus~E. Nebel, Raphael Reitzig, and Ulrich Laube.
\newblock {Engineering Java 7's Dual Pivot Quicksort Using MaLiJAn}.
\newblock In Peter Sanders and Norbert Zeh, editors, {\em ALENEX 2013}, pages
  55--69. SIAM, 2013.

\end{thebibliography}
